\documentclass[letterpaper]{article} 
\usepackage[draft]{aaai2026}  
\usepackage{times}  
\usepackage{helvet}  
\usepackage{courier}  
\usepackage[hyphens]{url}  
\usepackage{graphicx} 
\usepackage{natbib}  
\usepackage{caption} 
\usepackage{comment}
\usepackage{amsmath}
\usepackage{amsthm}
\usepackage{amssymb}
\usepackage{amsfonts}
\usepackage{dsfont}
\usepackage[linesnumbered,ruled,lined,longend]{algorithm2e}
\SetKwComment{Hline}{}{\vspace{-3mm}\textcolor{gray}{\hrule}\vspace{1mm}}
\SetKwInOut{Require}{Input}\SetKwInOut{Ensure}{Output}
\SetKwProg{Subr}{subroutine}{}{}
\usepackage{color}
\usepackage{xcolor}
\usepackage{booktabs}
\usepackage[switch]{lineno}

\usepackage{algorithm}
\usepackage{algorithmic}
\usepackage{newfloat}
\usepackage{listings}
\DeclareCaptionStyle{ruled}{labelfont=normalfont,labelsep=colon,strut=off} 
\floatstyle{ruled}
\newfloat{listing}{tb}{lst}{}
\floatname{listing}{Listing}
\newcommand{\Description}[2][]{}

\newcommand{\argmax}{\mathop{\rm arg~max}\limits}

\newtheorem{theorem}{Theorem}
\newtheorem{proposition}{Proposition}

\newcommand{\boldtypefaceseries}{\mdseries}

\title{On the Power of Perturbation under Sampling in Solving Extensive-Form Games}
\author{
Wataru Masaka$^1$,
Mitsuki Sakamoto$^2$,
Kenshi Abe$^{2,1}$,
Kaito Ariu$^2$,\\
Tuomas Sandholm$^{3}$,
Atsushi Iwasaki$^1$
}
\affiliations{

\textsuperscript{\rm 1}The University of Electro-Communications\\
\textsuperscript{\rm 2}Cyberagent, Inc.\\
\textsuperscript{\rm 3}Carnegie Mellon University, Strategy Robot, Inc., Strategic Machine, Inc., Optimized Markets, Inc.\\
m2110581@edu.cc.uec.ac.jp, atsushi.iwasaki@uec.ac.jp
}

\usepackage{bibentry}

\begin{document}

\maketitle

\begin{abstract}
\boldtypefaceseries
We investigate how perturbation does and does not improve the Follow-the-Regularized-Leader (FTRL) algorithm in solving imperfect-information extensive-form games under sampling, where payoffs are estimated from sampled trajectories. 
While optimistic algorithms are effective under full feedback, they often become unstable in the presence of sampling noise. 
Payoff perturbation offers a promising alternative for stabilizing learning and achieving \textit{last-iterate convergence}. 
We present a unified framework for \textit{Perturbed FTRL} algorithms and study two variants: PFTRL-KL (standard KL divergence) and PFTRL-RKL (Reverse KL divergence), the latter featuring an estimator with both unbiasedness and conditional zero variance.
While PFTRL-KL generally achieves equivalent or better performance across benchmark games, PFTRL-RKL consistently outperforms it in Leduc poker, whose structure is more asymmetric than the other games in a sense. 
Given the 
modest advantage of PFTRL-RKL, we design 
the second experiment to isolate the effect of conditional zero variance, 
showing that the variance-reduction property of RKL improve last-iterate performance.
\mdseries

\end{abstract}


\section{Introduction}
\label{sec:introduction}

Extensive-form games (EFGs) model sequential interactions among agents. Players make a sequence of decisions under imperfect information where they may not directly observe the actions by chance or the other player(s). 
Over the past decade, finding an (approximate) equilibrium in extensive-form games by using no-regret learning algorithms has been extensively studied~\citep{moravcik:science:2017,brown2018superhuman,brown2019solving}.
When each player minimizes regret, the time-averaged strategies approximate Nash equilibria in two-player zero-sum games, that is, the algorithms guarantee \textit{average-iterate} convergence. 
However, the actual sequence of strategies does not necessarily converge: it can cycle or even diverge~\citep{mertikopoulos2018cycles,bailey2018multiplicative}.
This is problematic as averaging can demand significant memory and computational resources for large games, e.g., if neural networks are used for state generalization, so each snapshot to be included in the averaging requires storing a separate neural network.

This is one of the main motivations for the study of {\it last-iterate} convergence, a stronger requirement  
where the strategies themselves converge to an equilibrium. 
%
However, the celebrated optimistic approach, e.g., \cite{daskalakis2018last}, faces difficulties when the feedback is contaminated by noise. Typically, each agent updates their strategy based on perfect gradient feedback from the payoff function at each iteration.
Noise arises and distorts this feedback when players cannot precisely observe the actions of others or when payoffs
must be estimated via merely \textit{sampling} the game-tree, as is the case in large games.
%
%
In such \textit{noisy feedback} settings, optimistic learning algorithms perform poorly, e.g.,~\citep{abe2023last}.
Optimistic algorithms, e.g., \citep{lee2021last}, 
which are designed for 
full-feedback settings requiring a complete traversal of the game-tree in each iteration, would similarly fail when payoff gradient estimation is needed, as in \textit{Monte Carlo Counterfactual Regret Minimization} (MCCFR)~\cite{lanctot2013monte}. 

\boldtypefaceseries
To this end, we examine how \emph{payoff perturbation}, a technique recently revived as a means to achieve \emph{last-iterate convergence}~\cite{perolat2021poincare,liu2022power,abe2023last,abe:icml:2024}, impacts FTRL dynamics in extensive-form games when feedback is derived via sampling. Rather than proposing a new state-of-the-art algorithm, our goal is to deepen theoretical and empirical understanding of how different forms of perturbation behave under outcome sampling---where payoffs are estimated from a single trajectory. This inquiry builds on classical insights from variational inequality theory~\cite{facchinei2003finite}, where strongly convex penalties stabilize learning, and connects to recent advances like \textit{Reward-Transformed FTRL}~\cite{perolat2021poincare}, the basis of \textit{DeepNash}~\cite{perolat:science:2022}.


We propose a unified framework for \emph{perturbed FTRL algorithms under outcome sampling}, which generalizes prior work and supports variance-sensitive analysis. Within this framework, we introduce \emph{PFTRL-RKL}, an extension of Mutant FTRL~\cite{abe2022mutationdriven, abe2023last} (originally developed for normal-form games) that modulates perturbation strength via \emph{Reverse KL divergence} from a fixed anchoring strategy. When the anchor keeps uniform, PFTRL-RKL coincides with \emph{replicator-mutator dynamics} from evolutionary game theory~\cite{Hofbauer1998}. We also revisit a KL-based variant, \emph{PFTRL-KL}, and show that it recovers Reward-Transformed FTRL. These formulations allow us to compare variance behavior under sampling and investigate their practical implications for convergence.


Although sampling reduces the computational cost of FTRL in one iteration, it often introduces high-variance estimates of payoffs~\cite{schmid2019variance,davis2020low}, which can destabilize learning. To address this, we analyze the variance of the two perturbation schemes. We show that the Reverse KL-based estimator used in PFTRL-RKL is both \textit{unbiased} and admits \textit{conditional zero variance} in the perturbation term---i.e., the variance of the cumulative perturbation magnitude is zero when conditioned on a sampled trajectory, even though the non-perturbed Q-value estimates exhibit non-zero variance. This makes PFTRL-RKL particularly appealing in high-variance sampling regimes.

In our experiments, both perturbed variants consistently stabilize last-iterate behavior and achieve significantly lower exploitability than standard FTRL. While PFTRL-KL generally matches or exceeds PFTRL-RKL on most benchmark games, we observe that PFTRL-RKL outperforms KL in Leduc poker, whose structure is more asymmetric than the other games in a sense.  
Given the relatively modest advantage of PFTRL-RKL, we design the second experiment to isolate the effect of conditional zero variance and confirms that variance reduction via Reverse KL can yield tangible improvements in last-iterate, especially when perturbation strength is larger.  
\mdseries

\subsection{Related Literature}
\label{sec:related lieterature}

Traditionally, finding equilibria in extensive-form games has involved transforming the problem into a linear program, but these methods do not scale to large-scale games~\citep{gilpin:jacm:2007}.
CFR~\citep{zinkevich2007regret} and its modern variants~\citep{tammelin2014solving,brown2019solving,farina:icml:2019,Farina21:Faster,Zhang24:Faster} have become foundational for solving extensive-form games. 
However, a full game-tree traversal is prohibitive in large games. 
MCCFR~\citep{lanctot2009nips,lanctot2013monte} addresses this by sampling parts of the game tree and estimating values (utility gradients or regrets) from sampled histories instead of traversing the entire tree. Two key unbiased estimators, external sampling and outcome sampling, have been proposed within this framework. This paper focuses on outcome sampling, which is computationally less expensive than external sampling but suffers from high variance due to its reliance on importance sampling to maintain unbiasedness.

Variance reduction has been widely studied in the reinforcement learning community to improve the performance of such estimators. 
For example, policy gradient algorithms commonly use baseline functions that depend on states to reduce variance. 
Building on this idea, \citet{schmid2019variance}
introduced variance-reduction techniques for MCCFR, and 
\citet{davis2020low}
extended baseline functions to depend on both states and actions. Alternatively, ESCHER~\cite{mcaleer:iclr:2023} introduced an unbiased estimator that avoids importance sampling, thereby reducing variance along sampled trajectories. Moreover, 
\citet{farina2021bandit} reformulated extensive-form games as bandit linear optimization, and subsequent work by \citet{kozuno2021learning} proposed low-variance estimators using implicit exploration, further refined with a balanced strategy~\cite{bai2022nearoptimal,fiegel2023adapting}. Notably, 
\citet{fiegel2023adapting} 
highlighted last-iterate convergence under sampling as a key challenge for future research, which the present paper seeks to address.

Despite the progress made by CFR-based methods, these approaches face challenges in requiring the computation of time-averaged strategies, which significantly increases computational complexity and memory usage~\citep{michael2015headsup}. 
To address this, recent research has focused on algorithms that achieve last-iterate convergence, where strategies themselves converge to equilibrium without relying on averaging. 
Optimistic algorithms  
have shown promise in achieving last-iterate convergence, in both normal-form~\citep{daskalakis2017training,daskalakis2018last,mertikopoulos2018optimistic,wei2020linear} and extensive-form~\citep{Farina21:Faster,lee2021last,liu2022power} games, albeit typically under full game-tree traversals. 
However, under sampling, optimistic algorithms have been shown to experimentally fail to converge to equilibrium even in normal-form games~\citep{abe2022mutationdriven,abe2023last}. 
Unlike the aforementioned prior studies, the present paper primarily investigates whether perturbed FTRLs can exhibit last-iterate convergence under sampling. Our goal is not to develop state-of-the-art algorithms, but to provide insights into the behavior of perturbed FTRLs in this challenging setting.

\section{Preliminaries}
\label{sec:preliminaries}


An extensive-form zero-sum game with imperfect information is defined as a tuple $\langle N, c, H, Z, A, \tau, \pi_c, u, \mathcal{X} \rangle$.
There are a finite set $N$ of players and a chance player $c$. 
$H=\bigcup_{N\cup\{c\}}H_i$ is the set of all possible {\it histories}, where each history is a sequence of {\it actions} and $H_i$ is the set of histories of player $i$'s action. 
We define $h \sqsubseteq h^{\prime}$ to mean $h\in H$ is a prefix sequence or equal to $h^{\prime}\in H$.
Assume {\it terminal histories} $Z_i \subset H_i$ for all $i$, we have $Z=\bigcup_{N\cup\{c\}}Z_i$ as the set of all terminal histories where the game has ended and the player has no available actions. 
At each history $h\in H\setminus Z$, the current player chooses an action $a\in A(h)$. We denote $A(h)$ as the set of actions available at history $h$ that lead to a successor history $(ha)\in H$. 
A \textit{player function} $\tau: H\setminus Z \to N \cup \{c\}$ maps each history $h$ to the player that chooses the next action at $h$. The chance player $c$ acts according to the defined distribution $\pi_c(\cdot | h)\in \Delta(A(h))$. 
A {\it payoff function} $u_i(h,a)$ maps each history $h\in H$ and action $a\in A(h)$ to a real value for player $i$. We assume that $u_i(h,a)=0$ holds if $ha\neq z\in Z$. Only the terminal histories produce non-zero payoffs. Also, if the game is two-player zero-sum, $u_1(h,a)=-u_2(h,a)$ holds. 


For each player $i\in N$, the collection of \textit{information sets} $X_i\in \mathcal{X}$ are \textit{information partitions} of the histories $\{h\in H | \tau(h)=i\}$. Player $i$ does not observe the true history $h$, but only the information set $x\in X_i$ corresponding to $h$. This implies that for each information set $x$, if any two histories $h, h^{\prime}$ belong to $x$, these histories are indistinguishable to the player $i$: $A(h)=A(h')$ for any $h, h'\in x$, which we then denote $A(x)$. 
We also denote $x(h)\in \mathcal{X}$ as an information set being reached by history $h$.


A \textit{strategy} or \textit{policy} $\pi_i(\cdot | x)$ maps an information set to a distribution over $\Delta(A(x))$. 
Each player $i$ chooses actions according to the strategy at each information set $x\in X_i$. When we restrict $\pi$ over $X_i$, we write a \textit{strategy profile} $\pi=(\pi_i,\pi_{-i})$


{\bf Reach Probabilities.} The \textit{reach probability} $\rho^{\pi}(h)$ of history $h$ under a strategy profile $\pi$ is 
\begin{align*}
\rho^{\pi}(h)=\prod_{(h^{\prime} a') \sqsubseteq h}\pi_{\tau(h^{\prime})}(a'|x(h^{\prime})).
\end{align*}
The reach probability $\rho^{\pi}(h)$ can be decomposed into $\rho_i^{\pi}(h)\rho_{-i}^{\pi}(h)$ for all $h\in H$ as the product of the reach probability of player~$i$ of history~$h$ and that of player~$-i$ (and the chance player $c$) of history~$h$. 
The probability of transitioning to $h^{\prime}$ from $h$ is given by:
\begin{align*}
\rho^{\pi}(h, h^{\prime}) =
\begin{cases}
\frac{\rho^{\pi}(h^{\prime})}{\rho^{\pi}(h)} & \mathrm{if} ~ \rho^\pi(h)>0 ~ \mathrm{and} 
~ h\sqsubseteq h^{\prime}, \\
0 & \mbox{otherwise.}
\end{cases}.
\end{align*}
The transition probability can also be decomposed into $\rho_i^{\pi}(h, h^{\prime})$ and $\rho_{-i}^{\pi}(h, h^{\prime})$. If history $h$ is the prefix sequence of $h'$, the transition probabilities are 
\begin{align*}
    \rho_i^{\pi}(h, h^{\prime})=\frac{\rho_i^{\pi}(h^{\prime})}{\rho_i^{\pi}(h)}\ \text{and}\ \rho_{-i}^{\pi}(h, h^{\prime})=\frac{\rho_{-i}^{\pi}(h^{\prime})}{\rho_{-i}^{\pi}(h)}
\end{align*} and these probabilities become zero, otherwise. 

Under \textit{perfect recall}~\cite{zinkevich2007regret}, in which the players do not forget any information that they once observed, let us define the reach probability of information set $x\in \mathcal{X}$, or summing over histories in $x$, under strategy profile $\pi$ as 
\begin{align*}
\rho^{\pi}(x)=\sum_{h\in x}\rho^{\pi}(h) = \rho_{i}^{\pi}(h)\left(\sum_{h'\in x}\rho_{-i}^{\pi}(h')\right), 
\end{align*}
for all $h\in x$. 
Also, we can write $\rho_{i}^{\pi}(h)$ as $\rho_{i}^{\pi}(x)$ for any $h\in x$, and $\sum_{h'\in x}\rho_{-i}^{\pi}(h')$ as $\rho_{-i}^{\pi}(x)$.

\textbf{Value Functions.} Given a strategy profile $\pi$, let us define the expected cumulative payoff for each player $i$ as 
\begin{align*}
u_i(\pi)=\sum_{h \in H \backslash Z } \sum_{a \in A(h)} \rho^{\pi}( ha)u_i(h,a).
\end{align*}
Next, let us define the Q-value of a policy $\pi$ for player $i$ at history $h$ while taking action $a$ by 
\begin{align}\label{eq:Q-value}
    q^{\pi}_i(h, a) = \sum_{h'a' \sqsupseteq ha} \rho^{\pi}(ha, h'a') u_i(h',a').
\end{align}
Finally, let us define {\it counterfactual value}, or value of a policy $\pi$ for player $i$ at information set $x$ while taking action $a$, as 
\begin{align}\label{eq:counterfactual value}
v_i^{\pi}(x, a)=\sum_{h \in x} \rho^{\pi}_{-i}(h) q_i^{\pi}(h,a).
\end{align}

\textbf{Nash Equilibrium.}
A popular solution concept for extensive-form games is a \textit{Nash equilibrium}~\citep{nash1951non}, where no player can increase their expected utility by deviating from their designated strategy. In two-player zero-sum extensive-form games, a Nash equilibrium $\pi^{\ast}=(\pi_1^{\ast}, \pi_2^{\ast})$ ensures the following condition: $\forall \pi_1\in \Sigma_1, \forall \pi_2\in \Sigma_2,$
\begin{align*}
    u_1(\pi_1^{\ast}, \pi_2) \geq u_1(\pi_1^{\ast}, \pi_2^{\ast}) \geq u_1(\pi_1, \pi_2^{\ast}),
\end{align*}
where $\Sigma_i$ symbolizes the set of all strategies for player $i$.
Furthermore, we define 
\begin{align*}
\mathrm{exploit}(\pi):=\max_{\widetilde{\pi}_1\in \Sigma_1}u_1(\widetilde{\pi}_1, \pi_2) + \max_{\widetilde{\pi}_2\in \Sigma_2}u_2(\pi_1, \widetilde{\pi}_2),
\end{align*}
as \textit{exploitability} of a given strategy profile $\pi$. Exploitability is a metric for measuring the closeness of $\pi$ to a Nash equilibrium $\pi^{\ast}$ in two-player zero-sum games \citep{johanson2011accelerating,lockhart2019computing,timbers2020approximate,abe2020off}.
By definition, exploitability satisfies $\mathrm{exploit}(\pi)\geq 0$ for any $\pi$, and it is $0$ if and only if $\pi$ is a Nash equilibrium.
At iteration $t$, for all information set $x$, a last-iterate strategy is denoted as $\pi^t(x)$
and an average-iterate strategy $\bar{\pi}^t(x)$ is denoted as $\frac{\sum_{s=1}^t \rho_i^{\pi^s}(x)\pi^s(x)}{\sum_{s=1}^t\rho_i^{\pi^s}(x)}$.


\textbf{Follow the Regularized Leader and Sampling.} We consider the online learning setting with a finite and discrete number of iterations. At each iteration $t\geq 1$, each player $i\in N$ determines her strategy $\pi_i^t$ based on the previously observed counterfactual values.
By traversing the entire game-tree or sampling a part of terminal histories $Y\subseteq Z$, each player $i$ calculates or estimates the counterfactual values $v_i^{\pi^t}(x,a)$ for each $x\in X_i$ and $a\in A(x)$.

This paper focuses on a widely used algorithm, FTRL. In imperfect information games, it defines a sequence of strategies $(\pi^t)_{t\in\{1,2,\ldots\}}$ for all $i\in N$ and $x\in \mathcal{X}$ as follows: 
\begin{align}\label{eq:ftrl}
    \pi^{t+1}_{i}(\cdot|x) \!=\! \argmax_{\pi\in \Delta(A(x))} \!\left\{  \eta\left\langle \sum_{s=1}^t v_i^{\pi^{s}}(x,\cdot), \pi\right\rangle \! - \! \psi_i(\pi)\right\},
\end{align}
where $\eta$ is a learning rate and $\psi_i: \Delta(A(x))\to \mathbb{R}$ is a strongly convex regularization function. 

To calculate the counterfactual values $v_i^{\pi^{s}}$ exactly, a prohibitive computational cost is required. 
Thus, to reduce the cost, some sampling schemes are used to estimate the counterfactual values by sampling a portion of terminal histories at each iteration.
%
%
We consider the following sampling scheme: 1) Let $\mathcal{Y}=\{Y_1, \cdots, Y_k\}$ be a set of subsets of the terminal histories $Z$, where the union of $\mathcal{Y}$ is $Z$; 2) At each iteration $t$, we sample one of the subsets $Y_j$ according to a \textit{sampling strategy}, or a predefined probability distribution $p\in \Delta(\mathcal{Y})$; 3) We estimate the counterfactual values from the sampled terminal histories $Y_j$ (formally defined later).
This sampling approach has been developed in the context of CFR
and initiated as 
MCCFR~\citep{lanctot2009nips}, which is an unbiased approximation of CFR and retains its desirable properties. 

\citet{lanctot2009nips} 
define two sampling schemes: \textit{external} and \textit{outcome} sampling. 
External sampling samples only the opponent’s (and chance’s) choices, requiring a forward model of the game to recursively traverse all subtrees under the player’s actions. In contrast, outcome sampling is the most extreme sampling variant, where blocks consist of a single terminal history. It is the only model-free variant of MCCFR that aligns with the standard reinforcement learning loop, allowing the agent to learn solely from its experience with the environment.
In the context of FTRL, outcome sampling is used to estimate the counterfactual values while representing strategies via neural network~\cite{perolat2021poincare,sokota2022unified}, although we focus on a tabular representation to make the effect of perturbation clear. 

This paper focuses on outcome sampling because it provides less information (less feedback) than external sampling. We expect external sampling to perform better than (or as well as) outcome sampling, which is usually used even in the context of FTRL~\cite{perolat2021poincare,sokota2022unified}. Please consult Supplementary Material~\ref{sec:external sampling} in detail. 
In addition, usually a current or last-iterate strategy is used to sample trajectories for estimating counterfactual values. 
However, 
\citet{mcaleer:iclr:2023} 
recently showed that, for CFR, sampling via a fixed strategy, which remains unchanged across iterations, outperforms sampling via a last-iterate strategy with some randomization (e.g., \(\varepsilon\)-greedy). Therefore, unless otherwise noted, this paper adopts a fixed sampling strategy, where actions are chosen at random in each iteration to estimate counterfactual values.
%

\section{Outcome Sampling Perturbed FTRLs}
\label{sec:os-pftrl}

\subsection{Perturbed FTRLs for Extensive-Form Games}
\label{sec:extenstion}

This section first presents Perturbed FTRL with \textit{Reverse Kullback-Leibler divergence} (PFTRL-RKL), which
%
perturbs each player's payoff function via the divergence $d_i^{\pi,\sigma}(h,a)$
when taking action $a$ at history $h$. 
It identifies the magnitude of the perturbation as the gradient of the KL divergence\footnote{$\mathrm{KL}(\pi_i, \pi_i')$ is defined as $\sum_{j}\pi_{ij}\ln \pi_{ij}/\pi_{ij}'$ for any two strategies $\pi_i$ and $\pi'_i$.} between the anchoring and current strategies $\nabla_{\pi_i}\mathrm{KL}(\sigma_i,\pi_i)$: 
\begin{align}\label{eq:reverse KL divergence}
    d_i^{\pi,\sigma}(h,a)=\frac{\mathds{1}_{i=\tau(h)}}{\pi_i(a|x(h))}(\sigma_i(a|x(h))-\pi_i(a|x(h)))
\end{align}
where $\sigma_i(\cdot\mid x)$ is a probability simplex over $A(x)$, $\mathds{1}$ is an indicator function. 
Note that the KL divergence 
typically takes the current strategy as the first argument and the anchoring strategy as the second, resulting in $d_i^{\pi,\sigma}(h,a)=\nabla_{\pi_i} \mathrm{KL}(\pi_i, \sigma_i)$. In contrast, we reverse the arguments and derive $d_i^{\pi,\sigma}(h,a)$ from $\nabla_{\pi_i} \mathrm{KL}(\sigma_i, \pi_i)$. Therefore, we refer to the KL divergence used herein as the Reverse KL divergence.

The cumulative magnitude of perturbation reaching history $ha$ under the current strategy $\pi$ is given by
\begin{align}\label{eq:cumulative magunitude of perturbation}
\delta_i^{\pi,\sigma}(h,a)=\sum_{h'a'\sqsupseteq ha}\rho^\pi(ha,h'a')d_i^{\pi,\sigma}(h',a').    
\end{align}
We then obtain the perturbed Q-value for each history-action pair by summing the non-perturbed Q-value in Eq.~\ref{eq:Q-value} and Eq.~\ref{eq:cumulative magunitude of perturbation} with the perturbation strength $\mu$:  
\begin{align*}
q_i^{\pi,\sigma}(h,a)=q_i^\pi(h,a)+\mu\delta_i^{\pi,\sigma}(h,a).
\end{align*}
Following the construction of Eq.~\ref{eq:counterfactual value}, 
we next construct the perturbed counterfactual value while choosing action $a$ at information state $x$ under the current strategy $\pi$ and the anchoring strategy $\sigma$: 
\begin{align}\label{eq:perturbed counterfactual value}
v_i^{\pi,\sigma}(x,a)=\sum_{h\in x}\rho_{-i}^\pi(h) q_i^{\pi,\sigma}(h,a).
\end{align}

Building on those components, we have PFTRL-RKL, which updates the strategy $\pi^t(\cdot\mid x)$ for each information set $x$ as follows.
{\small\begin{align}\label{eq:perturbed ftrl}
    \pi^{t+1}_{i}(\cdot|x) \!=\! \argmax_{\pi\in \Delta(A(x))} \!\left\{  \eta\left\langle \sum_{s=1}^t v_i^{\pi^{s}, \sigma}(x,\cdot), \pi\right\rangle \! - \! \psi_i(\pi)\right\}
\end{align}}
where $\eta$ is a learning rate and $\psi_i: \Delta(A(x))\to \mathbb{R}$ is a strongly convex regularization function. We here just replace $v_i^{\pi^{s}}(x,\cdot)$ in Eq.~\ref{eq:ftrl} with the perturbed counterfactual value $v_i^{\pi^{s}, \sigma}(x,\cdot)$ in Eq.~\ref{eq:perturbed ftrl}.

\boldtypefaceseries
Next we demonstrate that the update rule of PFTRL-RKL is interpreted as an extension of Mutant FTRL—originally developed for normal-form games—to the setting of extensive-form games under outcome sampling.
\mdseries
\begin{proposition}
    The update rule of PFTRL-RKL in Eq.~\ref{eq:perturbed ftrl}, whose magnitudes of perturbation are specified by Eq.~\ref{eq:reverse KL divergence}, is equivalent to Mutant FTRL~\cite{abe2022mutationdriven,abe2023last}.
%
\end{proposition}
\boldtypefaceseries
Mutant FTRL is inspired by the replicator-mutator dynamics in evolutionary game theory~\cite{Hofbauer1998,zagorsky:plosone:2013,Bauer:2019}, where strategy updates reflect both selection and post reproduction mutation. The strategy trajectory converges to a stationary point of the dynamics. This connection is one of the key motivations for adopting RKL in our formulation—a distinct contrast to conventional KL. 
%
\mdseries

Next we are going to identify PFTRL with the KL divergence (PFTRL-KL) by constructing $d_i^{\pi,\sigma}(h,a)$ with $\nabla_{\pi_i}KL(\pi_i,\sigma_i)$ which takes $\pi_i$ as the first argument and $\sigma_i$ as the second, resulting in  
\begin{align}\label{eq:KL divergence}
    d_i^{\pi, \sigma}(h, a) = \mathds{1}_{i=\tau(h)}\ln \frac{\sigma_i(a|x(h))}{\pi_i(a|x(h))}.
\end{align}

\begin{proposition}
    The update rule of PFTRL-KL in Eq.~\ref{eq:perturbed ftrl}, whose magnitudes of perturbation are specified by Eq.~\ref{eq:KL divergence}, is equivalent to Reward Transformed FTRL~\cite{perolat2021poincare}. 
\end{proposition}

We omit the proofs of the propositions, as they are straightforward, but we present the proof of the following theorems in the appendix.

\subsection{Estimators of Perturbed Counterfactual Value}
\label{sec:estimators}

Since calculating the counterfactual values exactly is computationally demanding, we utilize an outcome sampling scheme where the agent learns only from his or her experience.
We define the estimators under outcome sampling, i.e., \textit{outcome sampling PFTRL} (OS-PFTRL). If the perturbation strength $\mu$ is zero, this is equivalent to the conventional FTRL and we refer to that as OS-FTRL in Section~\ref{sec:experiments}. 
Let us construct an unbiased estimator $\widetilde{v}_i^{\pi, \sigma}$ for the counterfactual value $v_i^{\pi, \sigma}$ in Eq.~\ref{eq:perturbed counterfactual value}. 
We define the set of sampled histories in $Y_j$ as $H_j =\{h \in \mathcal{H} ~|~ h\sqsubseteq z \land z\in Y_j\}$. 
We denote the probability of reaching \( h \) as \( p(h) = \sum_{j : h \in H_j} p_j \), and the probability of transitioning from \( h \) to \( h' \) as \( p(h, h') \).
Then, when $Y_j$ is sampled, our estimators are defined as follows:
\begin{align}
\widetilde{v}_i^{\pi, \sigma}(x, a) &=  \sum_{h \in x} \frac{\rho^{\pi}_{-i}(h)}{p(h)} \left( \widetilde{q}_i^{\pi}(h, a) + \mu \widetilde{\delta}_i^{\pi, \sigma}(h, a) \right), 
\label{eq:estimator_expected_culmmutive_perturbed_payoff_pre_info} \\
\widetilde{q}_i^{\pi}(h, a) &= \sum_{h'a' \sqsupseteq ha \land h'a' \in H_j} \frac{\rho^{\pi}(ha,h'a')}{p(h, h'a')} u_i(h', a'), \label{eq:estimator_expected_culmmutive_original_payoff_pre_his} 
\\
\widetilde{\delta}_i^{\pi, \sigma}(h, a) &=  \mathds{1}_{h \in H_j} d^{\pi, \sigma}_i(h, a) \nonumber\\
&+ \sum_{h' \sqsupseteq ha \land h' \in H_j}\sum_{a' \in A(h')}  \frac{\rho^{\pi}(ha, h'a')}{p(h, h')} d^{\pi, \sigma}_i(h', a').
\label{eq:estimator_expected_culmmutive_payoff_pertabation_pre_his}
\end{align}

Eq.~\ref{eq:estimator_expected_culmmutive_perturbed_payoff_pre_info} is decomposed into two parts. The first one, defined in Eq.~\ref{eq:estimator_expected_culmmutive_original_payoff_pre_his}, corresponds to the standard Q-value in Eq.~\ref{eq:Q-value}. 
The second one, defined in Eq.~\ref{eq:estimator_expected_culmmutive_payoff_pertabation_pre_his}, corresponds to the cumulative magnitude of perturbation in Eq.~\ref{eq:cumulative magunitude of perturbation}. 
In Eq.~\ref{eq:estimator_expected_culmmutive_payoff_pertabation_pre_his}, the magnitude of perturbation $d_i^{\pi, \sigma}(h,a)$ of the unsampled actions at each sampled history need not to be estimated because the value can always be calculated with ease.
%
Given a sampling scheme, we update the strategy $\pi_i^t$ by the update rule in Eq.~\ref{eq:perturbed ftrl}, which uses $\widetilde{v}_i^{\pi^t, \sigma}$ instead of $v_i^{\pi^t, \sigma}$:
The entire pseudocode of our proposed OS-PFTRL is shown as Algorithm~\ref{alg:ftrl_perturbed} in Supplementary Material~\ref{sec:algorithms}. 

We now prove that our proposed estimator
for PFTRL-RKL, not PFTRL-KL, satisfies the two desirable properties of unbiasedness and conditional zero variance. That is, the estimator for the perturbed counterfactual value in Eq.~\ref{eq:estimator_expected_culmmutive_perturbed_payoff_pre_info} is unbiased, while the variance of the estimator for the cumulative magnitude of perturbation in Eq.~\ref{eq:cumulative magunitude of perturbation} is zero as long as $d_i^{\pi,\sigma}(h,a)$ is given by Eq.~\ref{eq:reverse KL divergence}, irrespective of which sampling scheme is used. 
On the other hand, if it is given by Eq.~\ref{eq:KL divergence}, the zero variance property in Eq.~\ref{eq:estimator_expected_culmmutive_payoff_pertabation_pre_his} no longer holds. 
Thus,  the variance of the estimated counterfactual values in PFTRL-RKL is smaller than that in KL. This suggests that OS-PFTRL-RKL should outperform OS-PFTRL-KL, a topic which we will discuss later in this paper. 

We first show that the estimator $\widetilde{v}_i^{\pi, \sigma}$ for the perturbed counterfactual value in Eq.~\ref{eq:estimator_expected_culmmutive_perturbed_payoff_pre_info} is unbiased.
\begin{theorem}\label{thm:unbiased}
For any $i\in N$, $x \in X_i$, and $a \in A(x)$, the perturbed counterfactual value estimator $\widetilde{v}_i^{\pi, \sigma}(x, a)$ satisfies
\begin{align*}
\mathbb{E}_{j \sim p_j}[\widetilde{v}_i^{\pi, \sigma}(x, a)] = v_{i}^{\pi, \sigma}(x,a)
\end{align*}
where history $H_j$ is sampled with probability $p_j$.
\end{theorem}

The next theorem shows that, conditioned on the event in which a history $h$ is sampled, the estimator $\widetilde{\delta}^{\pi, \sigma}_i(h, a)$ of the cumulative magnitude of perturbation in Eq.~\ref{eq:estimator_expected_culmmutive_payoff_pertabation_pre_his} has conditional zero variance.
\begin{theorem}
\label{thm:zero_variance}
Conditioned on the event where $Y_j$ is sampled, the value of the estimator $\widetilde{\delta}_i^{\pi, \sigma}(h, a)$ for any $h \in H_j$ and $a \in A(h)$ satisfies $\widetilde{\delta}^{\pi, \sigma}_i(h, a) = \delta^{\pi, \sigma}_i(h, a)$.
Therefore, we have for any $h\in H_j$ and $a \in A(h)$,
\begin{align*}
&\mathrm{Var}_{j\sim {p_j}}[\widetilde{\delta}_i^{\pi, \sigma}(h, a) ~|~ h \in H_j ] = 0. 
\end{align*}
\end{theorem} 
This theorem means that our estimator $\widetilde{\delta}^{\pi, \sigma}_i(h, a)$ can predict the exact value of the cumulative magnitude of perturbation $\delta^{\pi, \sigma}_i(h, a)$.
This is attributed to a unique feature of PFTRL-RKL, where the expected value is zero for any strategy $\pi$ and any history $h$:
\begin{align}
\mathbb{E}_{a\sim \pi_i(\cdot | x(h))}\left[\delta^{\pi, \sigma}_i(h, a)\right] = 0
\label{eq:zero_mean}    
\end{align}
as illustrated in Figure~\ref{fig:expected_perturbation_payoff}. 
All proofs are given in Supplementary Material~\ref{sec:appx_proof_of_proposed} due to space constraints. 

\begin{figure}[t]
    \centering
    \includegraphics[width=0.75\linewidth]{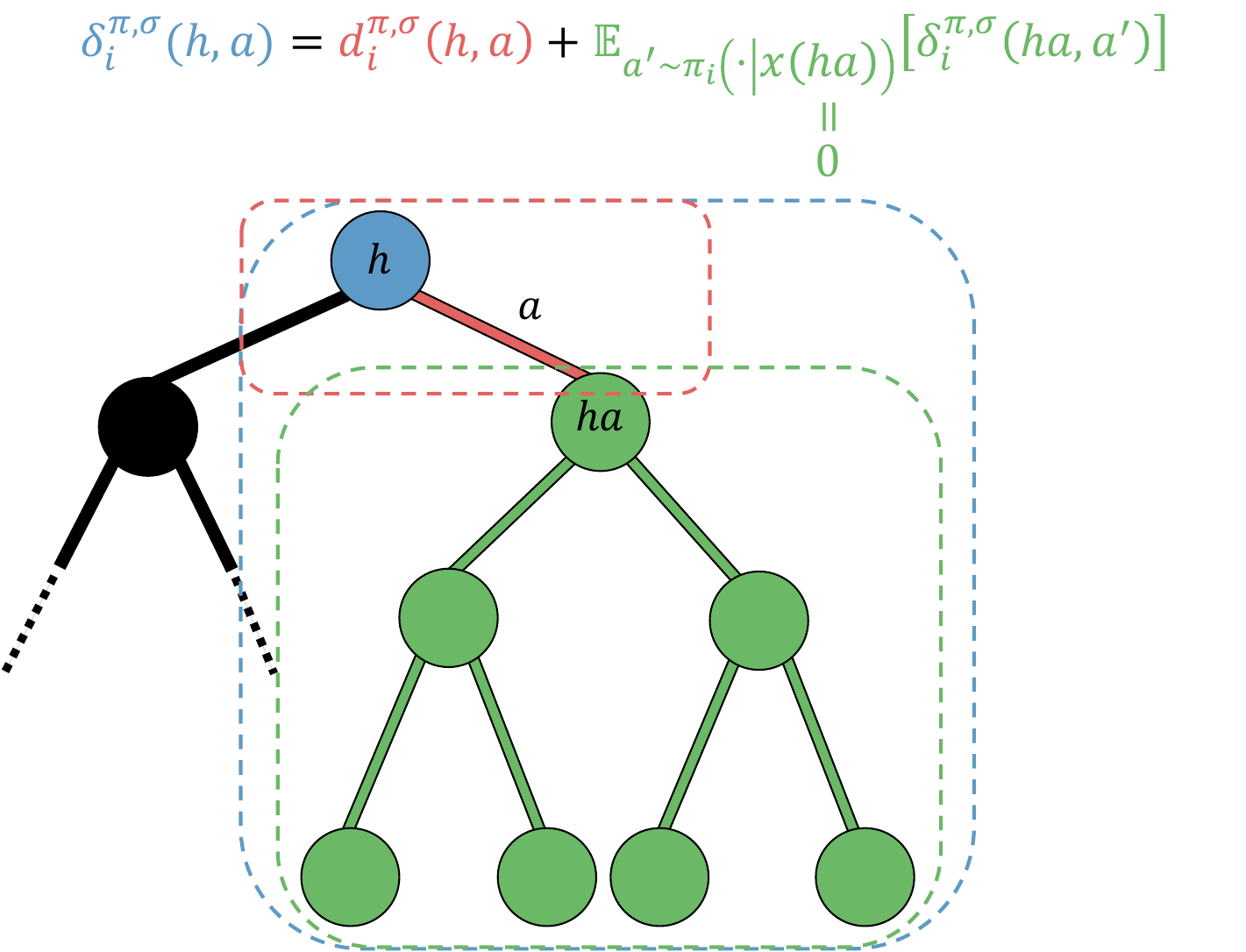}
    \caption{
    Illustration of the cumulative magnitude of perturbation $\delta_i^{\pi, \sigma}(h, a)$ with PFTRL-RKL payoff perturbation $d_i^{\pi, \sigma}$ at a history $h$ and an action $a$.
    $\delta_i^{\pi, \sigma}(h, a)$ is decomposed into two terms: 1) the immediate perturbation payoff (red color); 2) and the expected value of the cumulative magnitude of perturbation at a history $ha$,  $\mathbb{E}_{a'\sim \pi_i(\cdot | x(ha))}\left[\delta_i^{\pi, \sigma}(ha, a')\right]$ (green color). The figure highlights a key property  of PFTRL-RKL, where the expected value of the cumulative magnitude of perturbation is zero.
    }
    \label{fig:expected_perturbation_payoff}
    \Description[Illustration of the Q-value]{Illustration of the Q-value}
\end{figure}

This conditional zero-mean property enables us to simplify the perturbed counterfactual value estimator $\widetilde{v}_i^{\pi, \sigma}$. 
This property is desirable because it ensures that, given a sampled history, the perturbation magnitude is exactly known without additional randomness. 
We then describe how we represent the estimator under outcome sampling~\citep{lanctot2013monte}, which is also referred to as bandit feedback~\citep{kozuno2021learning,bai2022nearoptimal,fiegel2023adapting}. Since only a single terminal history is sampled according to a predefined sampling strategy $\pi'$ at each iteration, we have the set of the sampled histories is a singleton, that is, $|Y_j|=1$ for all $Y_j\in \mathcal{Y}$.

To sample histories, the $\varepsilon$-Greedy strategy is often adopted, which samples a history employing the current strategy with probability $1-\varepsilon$ and the uniform strategy that selects an action at random, otherwise. Formally, let us define $\pi_i'(a | x)$ as $(1 - \epsilon)\pi_i(a | x) + \frac{\epsilon}{|A(x)|}$ for any $x\in X_i$ and $a\in A(x)$ and $\pi_{-i}'$ as the current strategy of the opponents $-i$, $\pi_{-i}$.
However, it has recently been known that just the uniform strategy, i.e., setting $\varepsilon=1$, exhibits the advantage over the $\varepsilon$-Greedy strategy in the context of counterfactual regret minimization~\cite{mcaleer:iclr:2023}. 
We will discuss the difference in Section~\ref{sec:discussions} 
and, for simplicity, we adopt the uniform strategy to sample histories, unless noted. 

With Eq.~\ref{eq:zero_mean}, given a sampling strategy $\pi'$, let us rewrite the perturbed counterfactual value estimator $\widetilde{v}_{i}^{\pi, \sigma}(x, a)=$
\begin{align*}
& \!\!\sum_{h \in x \cap H_j} \frac{1}{\rho_{i}^{\pi'}(ha)} \!\sum_{h'a' \sqsupseteq ha \land h'a' \in H_j} \!\frac{\rho^{\pi}_{i}(ha, h'a')}{\rho^{\pi'}_{i}(ha, h'a')}u_i(h',a') \\
&\phantom{=} + \mu\sum_{h \in x \cap H_j} \frac{1}{\rho^{\pi'}_{i}(h)\pi_i(a|x)}\left(\sigma_i(a|x) - \pi_i(a|x)\right). 
\end{align*}
Algorithm~\ref{alg:outcome sampling} in Supplementary Material~\ref{sec:algorithms} illustrates the procedure for estimation of $\widetilde{v}_i^{\pi, \sigma}$ under outcome sampling.

In addition, as a benchmark, let us consider \textit{full walk}, or the full game-tree traversal scheme, where we sample all the terminal histories at each iteration, i.e., $\mathcal{Y}=\{Z\}$.
In this scheme, the perturbed counterfactual value estimator $\widetilde{v}_i^{\pi, \sigma}$ is equivalent to the exact perturbed counterfactual value $v_i^{\pi, \sigma}$.
Eq.~\ref{eq:zero_mean} enables us to simplify $\widetilde{v}_i^{\pi, \sigma}(x, a)$ for any $x\in X_i$ and $a\in A(x)$ and we obtain $\widetilde{v}_i^{\pi, \sigma}(x, a) = v_i^{\pi, \sigma}(x, a)=$
\begin{align*}
& \sum_{h \in x} \rho^{\pi}_{-i}(h) \left(q_i^{\pi}(h,a) + \frac{\mu}{\pi_i(a|x)}\left(\sigma_i(a | x) - \pi_i(a | x)\right) \right).
\end{align*}
We provide the pseudocode for calculating $v_i^{\pi, \sigma}$ by utilizing this equation as Algorithm~\ref{alg:full_information} in Supplementary Material~\ref{sec:algorithms}.

\begin{figure}[tb]
    \centering
    \includegraphics[width=0.96\linewidth]{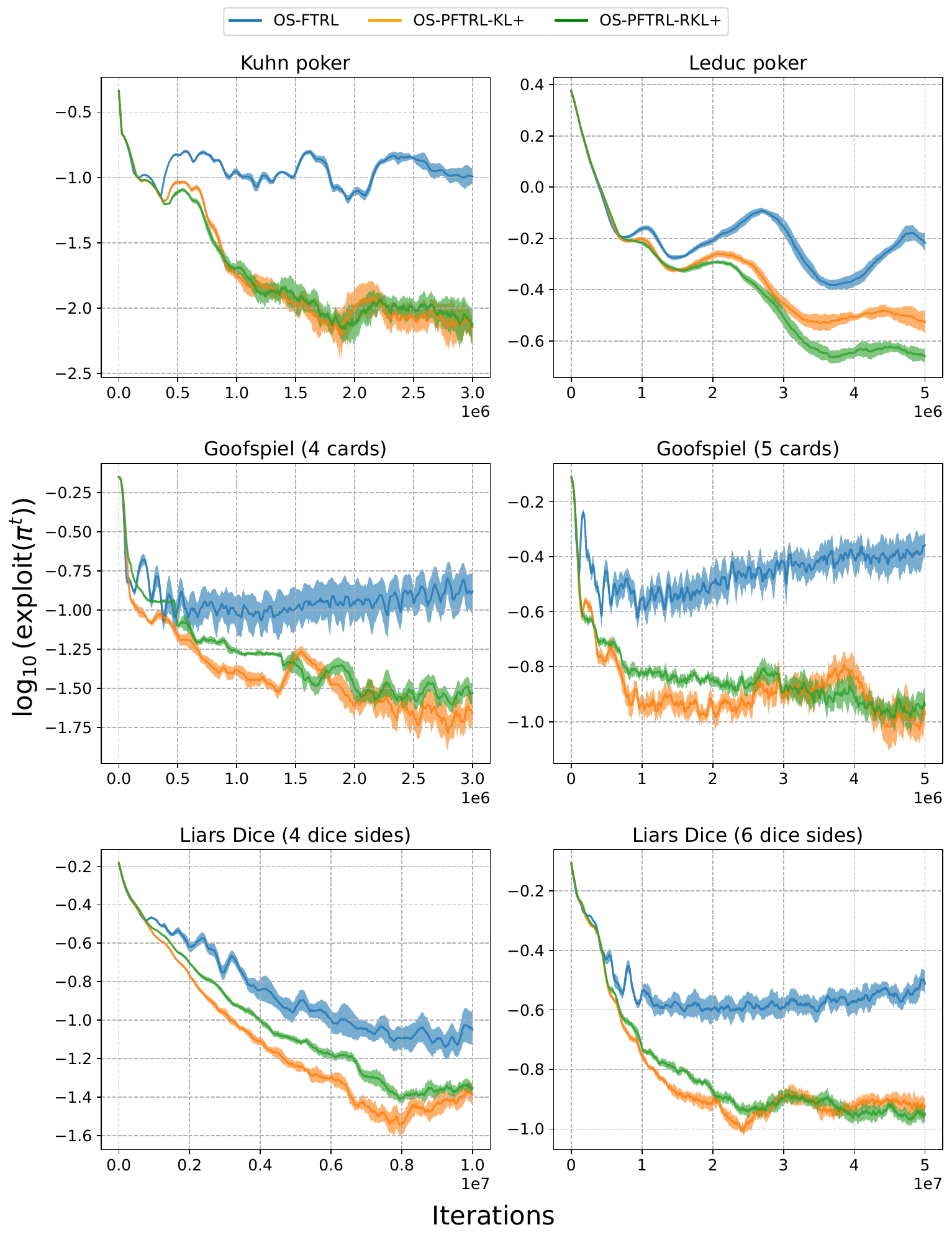}
    \caption{Exploitability of last-iterate $\pi^t$ under outcome sampling.}
    \label{fig:exploitability-last-iterate}
    \Description[Exploitability]{Exploitability of last-iterate $\pi^t$}
\end{figure}


\section{Experiments}
\label{sec:experiments}

In this section we investigate how well the FTRL-based algorithms 
perform on the common benchmark games: Kuhn poker~\citep{kuhn1951simplified}, Leduc poker~\citep{southey2005bayes}, Goofspiel~\citep{lanctot2013monte}, and Liar's Dice~\citep{ferguson1991models}. 
The number of information sets for each game is as follows: Kuhn poker has $12$, Leduc poker has $936$, Goofspiel with $4$ cards has $162$ and with $5$ cards has $2,124$, while Liar's Dice with $4$ dice sides has $1,024$ and with $6$ dice sides has $24,576$. We used the OpenSpiel~\citep{LanctotEtAl2019OpenSpiel} framework for our experiments. 

We built the \textit{anchoring strategy update}~\cite{perolat2021poincare,abe2023last,abe:icml:2024} into perturbed FTRL under sampling. Every time information set $x$ is visited $T_{\sigma} \leq T$ times, we replaced 
the anchoring strategy $\sigma_i(\cdot | x)$ for each $x$ with the current strategy $\pi_i^t(\cdot | x)$. We refer to the algorithms with this procedure as PFTRL-RKL+ and -KL+, respectively. The details are in Supplementary Material~\ref{sec:anchoring strategy updates}. 

We used the constant learning rate $\eta=0.0001$ for outcome sampling and the perturbation strength $\mu=0.1$ for perturbed FTRL. We initialize the anchoring strategy 
uniformly: 
$\sigma_{i}(\cdot | x) = (1/|A(x)|)_{a \in A(x)}$ for each information set $x$.
It is updated every $T_{\sigma}=100,000$ visits under outcome sampling. Assume that the sampling strategy is uniform~\cite{mcaleer:iclr:2023}. The exploitability is averaged across $10$ random seeds for each algorithm and is presented on a logarithmic scale.
We use the entropy regularizer $\psi_i(\pi(\cdot | x)) = \sum_{a\in A(x)}\pi(a|x)\ln \pi(a|x)$ in all experiments.

\begin{figure*}[tb]
    \centering
    \includegraphics[width=0.89\linewidth]{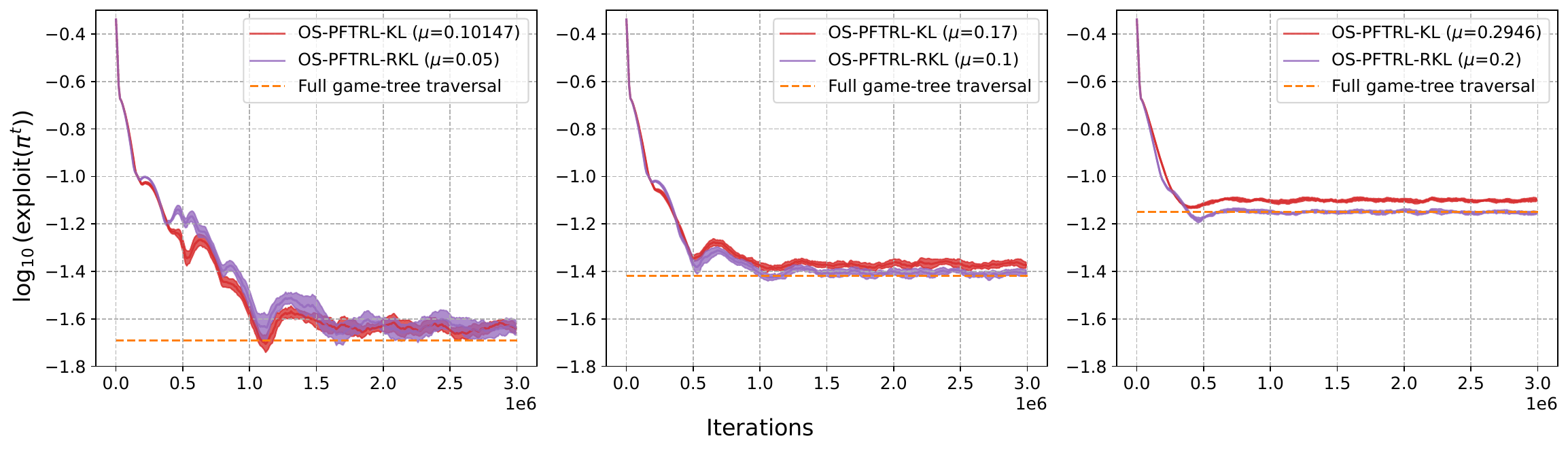}
    \caption{Exploitability difference of last-iterate between OS-PFTRL-RKL and -KL in Kuhn poker, with varying tuned perturbation strengths.}
    \label{fig:effect of perturabtion strength}
    \Description[Effect of perturbation strength]{Effect of perturbation strength}
\end{figure*}

\boldtypefaceseries
Figure~\ref{fig:exploitability-last-iterate} shows the exploitability of the last-iterate strategies $\pi_t$ under outcome sampling.\footnote{Results under full game-tree traversal and external sampling are presented in Supplementary Materials C and D.} Both perturbed variants, OS-PFTRL-KL+ and 
-RKL+, consistently outperform the unperturbed OS-FTRL, exhibiting faster convergence and lower exploitability. Notably, OS-PFTRL-RKL+ achieves a significant advantage in Leduc poker, where it outperforms KL+ in terms of both convergence rate and final exploitability. While KL+ modestly outperforms RKL+ in most of the other games, RKL+ consistently delivers stable results with smooth convergence behavior. These results suggest that payoff perturbation not only stabilizes FTRL dynamics but also enhances last-iterate performance in imperfect-information games, even under sampling, consistent with prior findings~\cite{perolat2021poincare,sokota2022unified}. However, that perturbation often performs poorly with regard to the average-iterate sense in Figure~\ref{fig:average iterate OS} of Supplementary Material~\ref{sec:comparison with cfr/cfr+}. 

To more precisely isolate the benefits of PFTRL-RKL, we design a controlled comparison that highlights the effect of its conditional zero variance property; the result is illustrated in Figure~\ref{fig:effect of perturabtion strength}. This procedure consists of three steps:
\begin{enumerate}
    \item We first fix a perturbation strength (e.g., $\mu = 0.05$) and run PFTRL-RKL using full game-tree traversals (the orange dashed line), observing its converged exploitability (e.g.,~$10^{-1.7}$);
    \item We then tune the perturbation strength for PFTRL-KL so that it reaches comparable exploitability under the same full-traversal setting. For example, PFTRL-KL with $\mu =0.010147$ yields the same level of converged exploitability as PFTRL-RKL with $\mu=0.05$;  
    \item  Using these matched settings, we switch to outcome sampling and examine whether OS-PFTRL-RKL (the purple line) still maintains an advantage over OS-PFTRL-KL (the red line).
\end{enumerate}
Since the perturbation magnitudes are calibrated for equivalent baseline performance, any observed differences under sampling can be attributed to the form of the perturbation, rather than its strength. 

Figure~\ref{fig:effect of perturabtion strength} presents the exploitabilities of OS-PFTRL-KL (the red line) and -RKL (the purple line) in Kuhn poker under outcome sampling, while presents that under full game-tree traversals (the orange dashed line), with varying perturbation strengths of OS-PFTRL-RKL ($\mu\in\{0.05, 0,1, 0.2\}$). The x-axis represents the number of iterations, and the y-axis measures exploitability on a logarithmic scale. Note that we do not apply the anchoring strategy update and it is expected that 
those three lines (oragen, purple, and red)  
do not achieve zero exploitability.

As the perturbation strength increases ($\mu \in {0.05, 0.1, 0.2}$ for RKL), a clear pattern emerges. At the lowest setting, both variants perform similarly under sampling. At moderate strength, RKL begins to outperform KL modestly. At the highest setting, PFTRL-RKL achieves significantly lower exploitability after 500{,}000 iterations.
%
Overall, this experiment reveals that PFTRL-RKL’s conditional zero variance property can lead to more robust convergence under sampling, especially when perturbation strength is larger. 
\mdseries

\section{Discussion}
\label{sec:discussions}
It is important to note that implementing optimistic variants under outcome sampling is not straightforward. Simply introducing optimism, as in Optimistic FTRL or Dilated OMWU~\cite{lee2021last}, is insufficient because the prediction vector becomes contaminated by noise or errors arising from estimates of the payoffs%
~\cite{abe2023last}.
While it remains an open question whether combining optimism with perturbation, such as in Reg-DOMWU and Reg-DOGDA~\cite{liu2022power}, is effective, designing an estimator for the prediction vector is nontrivial, making  implementation more complicated. Nonetheless, exploring whether perturbation can make optimistic variants feasible under sampling is a promising direction for future work. 

\boldtypefaceseries
Another important avenue for future research involves the modern CFR variants~\cite{brown2019solving,Farina21:Faster,Zhang24:Faster} that are the current state of the art 
in the average-iterate sense, under full game-tree traversals. Meanwhile, under outcome sampling, 
ESCHER~\cite{mcaleer:iclr:2023} have shown strong performance.
Our experiments using tabular representations in Supplementary Material~\ref{sec:comparison with cfr/cfr+} suggest that, in the last-iterate sense, our perturbation-based algorithms outperform the standard CFR and generally outperform CFR+ across most games. However, intriguingly, CFR+ performs the best in Leduc poker and Liars dice. 

In addition, perturbed variants of CFR have been proposed, such as \textit{Reg-CFR}~\cite{liu2022power} and \textit{Reward Transformation CFR+}~\cite{meng::2023}. We implement the latter as \textit{Perturbed CFR+-L2+} (PCFR+-L2+), which is the current state of the art under full game-tree traversal in the last-iterate sense. 
In the average-iterate sense, as we see in FTRL, perturbation does not improve CFR+. 
However, in the last-iterate sense, perturbation improve CFR+, since OS-PCFR+-L2+ consistently outperforms OS-CFR+. Nevertheless, it is outperformed by OS-FTRL-RKL+, -KL+, or both,  across most games except Liars Dice with 6 dice sides. The details are shown in Supplementary Material~\ref{sec:comparison with perturbed cfr+}. 
%
It remains an open and important question how perturbation can be effectively integrated with CFR-based algorithms to attain better last-iterate convergence.
\mdseries

\section{Conclusions}
\label{sec:conclusion}
We investigated how perturbing payoffs affects FTRL-based algorithms under outcome sampling in solving extensive-form, imperfect-information games. 
We devised a simple one-line modification to PFTRL-KL, which we coined PFTRL-RKL, whose theoretical variance in the estimator for counterfactual values is smaller than that of PFTRL-KL. We empirically showed that perturbation consistently improves performance in the last-iterate sense, which reduces memory and computational requirements compared to having to average the iterates. 


\bibliographystyle{aaai2026}
\bibliography{references}

\appendix
\onecolumn
\section{Proofs for Section~\ref{sec:os-pftrl}}
\label{sec:appx_proof_of_proposed}
\begin{proof}[Proof of Theorem \ref{thm:unbiased}]
\label{prf:unbiased}
\begin{align*}
&\mathbb{E}_{j \sim p_j}[\widetilde{v}_i^{\pi, \sigma}(x, a; j)]  \\ 
&= \sum_{j} p_j\sum_{h \in x}\frac{\rho^\pi_{-i}(h)}{p(h)}\tilde q_i^{\pi,\sigma}(h,a) \\
&= \sum_{j} p_j \sum_{h \in x} \frac{\rho^\pi_{-i}(h)}{p(h)} \left\{ \tilde q_i^\pi(h,a) + \mu\tilde\delta^{\pi,\sigma}_i(h,a) \right\} \\
    &\begin{aligned}
    = \sum_{j} p_j \sum_{h \in x}\frac{\rho^\pi_{-i}(h)}{p(h)} &\left\{ \sum_{h'a'\sqsupseteq ha\wedge h'a'\in H_j}\frac{\rho^\pi(ha,h'a')}{p(h,h'a')}u_i(h',a') + \mu\mathds{1}_{h \in H_j}d_i^{\pi,\sigma}(h,a)  \right.\\
    &\left.  + \mu\sum_{h'\sqsupseteq ha\land h'\in H_j}\sum_{a'\in A(h')}\frac{\rho^\pi(ha,h'a')}{p(h,h')}d_i^{\pi,\sigma}(h',a') \right\}
    \end{aligned} \\
&= \sum_{j}p_j \sum_{h \in x} \rho^{\pi}_{-i}(h) \left\{   \sum_{h'a' \sqsupseteq ha \land h'a' \in H_j} \frac{\rho^{\pi}(ha,h'a')}{p(h'a')} u_i(h', a')+ \mu \frac{\mathds{1}_{h \in H_j}}{p(h)} d^{\pi, \sigma}_i(h, a) \right. \\
& \left. \phantom{\sum_{j}p_j \sum_{h \in x} \rho^{\pi}_{-i}(h)=} + \mu \sum_{h' \sqsupseteq ha \land h' \in H_j} \sum_{a' \in A(h')} \frac{\rho^{\pi}(ha,h'a')}{p(h')} d^{\pi, \sigma}_i(h', a') \right\}  \\
&=\sum_{h \in x} \rho^{\pi}_{-i}(h) \sum_{h'a' \sqsupseteq ha } \frac{\rho^{\pi}(ha,h'a')}{p(h'a')} u_i(h', a') \sum_{j: h'a' \in H_j}p_j \\
&\phantom{\sum_{j}p_j \sum_{h \in x} \rho^{\pi}_{-i}(h)=}+ \mu \sum_{h \in x} \frac{\rho^{\pi}_{-i}(h)}{p(h)} d^{\pi, \sigma}_i(h, a) \sum_{j: h \in H_j}p_j
+ \mu \sum_{h \in x} \rho^{\pi}_{-i}(h) \sum_{h' \sqsupseteq ha }  \sum_{a' \in A(h')} \frac{\rho^{\pi}(ha,h'a')}{p(h')} d^{\pi, \sigma}_i(h', a') \sum_{j: h' \in H_j}p_j \\
&= \sum_{h \in x} \rho^{\pi}_{-i}(h) \sum_{h'a' \sqsupseteq ha} \rho^{\pi}(ha,h'a') u_i(h', a') 
+ \mu \sum_{h \in x} \rho^{\pi}_{-i}(h) \left( d^{\pi, \sigma}_i(h, a) + \sum_{h' \sqsupseteq ha} \sum_{a' \in A(h')} \rho^{\pi}(ha,h'a') d^{\pi, \sigma}_i(h', a')  \right) \\
&= \sum_{h \in x} \rho^{\pi}_{-i}(h) q_i^\pi(h,a) + \mu\sum_{h\in x}\rho^\pi_{-i}(h)\delta_{i}^{\pi, \sigma}(h,a) \\
&= \sum_{h \in x} \rho_{-i}^\pi(h)\left( q_{i}^\pi(h,a) + \mu\delta_{i}^{\pi,\sigma}(h,a)\right) \\
&= \sum_{h \in x}\rho_{-i}^{\pi}(h) q_i^{\pi, \sigma}(h,a) \\
\\ &= v_{i}^{\pi, \sigma}(x,a).
\end{align*}
\end{proof}

\begin{proof}[Proof of Theorem \ref{thm:zero_variance}]
\label{prf:zero_variance}

First, since $\sum_{a'\in A(h')}\sigma_{\tau(h')}(a'|x(h'))=1$ and $\sum_{a'\in A(h')}\pi_{\tau(h')}(a'|x(h'))=1$ for any $h'\in H\setminus Z$, 
the cumulative magnitude of perturbation 
$\delta_i^{\pi, \sigma}(h, a)$ can be rewritten as:
\begin{align*}
\delta^{\pi, \sigma}_i(h, a)  &= \sum_{h'a' \sqsupseteq ha} \rho^{\pi}(ha, h'a') d_i^{\pi, \sigma}(h', a') \\
&= \sum_{h'a' \sqsupseteq ha}\rho^{\pi}(ha,h'a') 
 \left( \frac{\mathds{1}_{i=\tau(h')}}{\pi_{\tau(h')}(a'|x(h'))}\left(\sigma_{\tau(h')}(a'|x(h')) - \pi_{\tau(h')}(a'|x(h'))\right)  \right) \\
&= \frac{1}{\pi_{i}(a|x(h))} \left(\sigma_{i}(a|x(h)) - \pi_{i}(a|x(h))\right)
+ \sum_{h' \sqsupseteq ha} \rho^{\pi}(ha, h') \mathds{1}_{i=\tau(h')}\sum_{a' \in A(h')} 
\left(\sigma_{\tau(h')}(a'|x(h')) - \pi_{\tau(h')}(a'|x(h'))\right) \\
&= \frac{1}{\pi_{i}(a|x(h))} \left(\sigma_{i}(a|x(h)) - \pi_{i}(a|x(h))\right).
\end{align*}

On the other hand, from the definition of the estimator $\widetilde{\delta}_i^{\pi, \sigma}(h, a;j)$, we have for any $h\in H$ and $a \in A(h)$:
\begin{align*}
\widetilde{\delta}_i^{\pi, \sigma}(h, a;j) &= \mathds{1} [h \in H_j]d^{\pi, \sigma}_i(h, a) + \sum_{h' \sqsupseteq ha \land h' \in H_j}\sum_{a' \in A(h')}  \frac{\rho^{\pi}(ha,h'a')}{p(h, h')} d^{\pi, \sigma}_i(h', a') \\
&=  \mathds{1} [h \in H_j] \frac{1}{\pi_{i}(a|x(h))} \left(\sigma_{i}(a|x(h)) - \pi_{i}(a|x(h))\right)  \\
&\phantom{=\mathds{1}_{h \in H_j}} + \sum_{h' \sqsupseteq ha \land h' \in H_j}\sum_{a' \in A(h')}  \frac{\rho^{\pi}(ha,h'a')}{p(h, h')}  
 \left( \frac{\mathds{1}_{i=\tau(h')}}{\pi_{\tau(h')}(a'|x(h'))}\left(\sigma_{\tau(h')}(a'|x(h')) - \pi_{\tau(h')}(a'|x(h'))\right)  \right) \\
&= \mathds{1} [h \in H_j] \frac{1}{\pi_{i}(a|x(h))} \left(\sigma_{i}(a|x(h)) - \pi_{i}(a|x(h))\right)  \\
&\phantom{=\mathds{1}_{h \in H_j}} + \sum_{h' \sqsupseteq ha \land h' \in H_j} \frac{\rho^{\pi}(ha,h')}{p(h, h')} \mathds{1}_{i=\tau(h')}\sum_{a' \in A(h')} \left(\sigma_{\tau(h')}(a'|x(h')) - \pi_{\tau(h')}(a'|x(h'))\right) \\
&= \mathds{1}_{h \in H_j}\frac{1}{\pi_{i}(a|x(h))} \left(\sigma_{i}(a|x(h)) - \pi_{i}(a|x(h))\right) .
\end{align*}
Therefore, under the event where $Q_j$ is sampled, we have for any $h \in H_j$ and $a \in A(h)$:
\begin{align*}
\widetilde{\delta}_i^{\pi, \sigma}(h, a;j) = \frac{1}{\pi_{i}(a|x(h))} \left(\sigma_{i}(a|x(h)) - \pi_{i}(a|x(h))\right)
\end{align*}

Hence, combining these, for any $h \in H_j$ and $a \in A(h)$:
\begin{align*}
&\widetilde{\delta}^{\pi, \sigma}_i(h, a; j) = \delta^{\pi, \sigma}_i(h, a).
\end{align*}
In other words,
\begin{align*}
&\mathrm{Var}_{j\sim {p_j}}[\tilde \delta_i^{\pi, \sigma}(h, a; j) ~|~ h \in H_j ] = 0. \\
\end{align*}
\end{proof}

\section{Algorithms}
\label{sec:algorithms}

This section details the proposed algorithms in this paper. Algorithm~\ref{alg:ftrl_perturbed} outlines the complete learning process for the variants of PFTRL, namely PFTRL-RKL+ and PFTRL-KL+. If the update interval \( T_{\sigma} \) is set to \(\infty\), the anchoring strategy remains fixed, corresponding to the case where no anchoring strategy update occurs, i.e., PFTRL-RKL and PFTRL-KL. 
Algorithm~\ref{alg:full_information} is used to compute the counterfactual value \( v_{i}^{\pi^{t}, \sigma} \) under full game-tree traversals, while Algorithm~\ref{alg:outcome sampling} estimates the counterfactual value \( \widetilde{v}_{i}^{\pi^{t}, \sigma} \) under outcome sampling.
Note that Algorithm~\ref{alg:outcome sampling} utilizes $\varepsilon$-Greedy to sample trajectories. According to ESCHER~\cite{mcaleer:iclr:2023}, unless noted otherwise, we set $\varepsilon$ to one throughout the experiments in this paper.

\begin{figure}[t!]
\makeatletter\let\@latex@error\@gobble\makeatother
\begin{algorithm}[H]\small
\caption{\normalfont{PFTRL-RKL+ and -KL+.}}
\label{alg:ftrl_perturbed}
\DontPrintSemicolon
\Require{Time horizon $T$, learning rate $\eta$, mutation parameter $\mu$, update interval $T_{\sigma}$}
$\pi_i^1(\cdot | x) \gets \left(\frac{1}{|A(x)|}\right)_{a\in A(x)}$ for all $i\in N$ and $x\in X_i$\;
$\sigma \gets \pi^1$\;
$\kappa_i[x] \gets 0$ for all $i\in N$ and $x\in X_i$\;
\For{$t=1,2\cdots,T$}{
    \For{$i\in N$}{
        Estimate the perturbed counterfactual values $v_{i}^{\pi^{t}, \sigma}$ by Algorithms \ref{alg:full_information} 
        or \ref{alg:outcome sampling} with inputs $(\pi^t, i, \mu, \sigma)$\;
        Let $X_{i, \mathrm{vst}}^t$ be the player $i$'s information sets visited at iteration $t$\;
        \For{$x\in X_{i, \mathrm{vst}}^t$}{
            Update the strategy by $\pi^{t+1}_{i}(\cdot|x) = \argmax_{\pi\in \Delta(A(x))} \left\{\eta\left\langle \sum_{s=1}^t v_i^{\pi^{s}, \sigma}(x,\cdot), \pi\right\rangle - \psi_i(\pi)\right\}$\;
            $\kappa_i[x] \gets \kappa_i[x] + 1$;\;
            \If{$\kappa_i[x] = T_{\sigma}$}{
                $\sigma_i(\cdot|x) \gets \pi_i^{t+1}(\cdot|x)$\;
                $\kappa_i[x] \gets 0$\;
            }
        }
    }
}
\end{algorithm}
    \Description[PFTRL-RKL+ and -KL+.]{PFTRL-RKL+ and -KL+.}
\end{figure}

\begin{figure}[t!]
\makeatletter\let\@latex@error\@gobble\makeatother
\begin{algorithm}[H]
\caption{\normalfont$\textsc{ValueCompute}(\pi, i, \mu, \sigma)$ for full game-tree traversals.}
\label{alg:full_information}
\DontPrintSemicolon
$v_i^{\pi, \sigma}(x, a) \gets 0$ for all $x\in X_i$ and $a\in A(x)$\;

\Subr{\normalfont$\textsc{Traverse}\left(h, i, \rho_{-i}\right)$}{
    \uIf{$h\in Z$} {
        \Return $0$\;
    }
    \uElseIf{$\tau(h)= c$}{
        \Return $\sum_{a\in A(h)}\pi_c(a | h) \cdot \left(\textsc{Traverse}\left(ha, i, \pi_c(a | h)\cdot\rho_{-i}\right) + u_i(h, a)\right)$\;
    }
    \ElseIf{$\tau(h)\neq i$}{
        \Return $\sum_{a\in A(x)}\pi_{\tau(h)}(a | x(h)) \cdot \left(\textsc{Traverse}\left(ha, i, \pi_{\tau(h)}(a | x(h))\cdot\rho_{-i}\right) + u_i(h, a)\right)$\;
    }
    Let $x$ be the information set containing $h$\;
    \uIf{$\tau(h)=i$}{
        $q_i[h] \gets 0$\;
        $q_i[h, a] \gets 0$ for all $a\in A(x)$\;
        \For{$a\in A(x)$}{
            $q_i[h, a] \gets \textsc{Traverse}(ha, i, \rho_{-i}) + u_i(h, a) + \mu d_i^{\pi,\sigma}(h, a)$\;
            $v_i^{\pi, \sigma}(x, a) \gets v_i^{\pi, \sigma}(x, a) + \rho_{-i} \cdot q_i[h, a]$\;
            $q_i[h] \gets q_i[h] + \pi_i(a|x) \cdot q_i[h, a]$\;
        }
        \Return $q_i[h]$\;
    }
}
\Hline{}
$\textsc{Traverse}(\emptyset, i, 1)$\;
\Return $v_i^{\pi, \sigma}$\;
\end{algorithm}
    \Description[ValueCompute]{ValueCompute}
\end{figure}

\begin{figure}
\makeatletter\let\@latex@error\@gobble\makeatother
\begin{algorithm}[H]
\caption{\normalfont$\textsc{ValueEstimateOS}(\pi, i, \mu, \sigma)$ for outcome sampling.}
\label{alg:outcome sampling}
\DontPrintSemicolon
$\widetilde{v}_{i}^{\pi, \sigma}(x, a) \gets 0$ for all $x\in X_i$ and $a\in A(x)$\;
\Subr{\normalfont$\textsc{TraverseOS}\left(h, i, \rho_{i}\right)$}{
    \uIf{$h\in Z$}{
      \Return $0$\;
    }
    \ElseIf{$\tau(h)=c$}{
      Sample action $a \sim \pi_c(a|h)$\;
      \Return $\textsc{TraverseOS}\left(ha, i, \rho_{i}\right) + u_{i}(h, a)$\;
    }
    Let $x$ be the information set containing $h$\;
    \uIf{$\tau(h)=i$}{
        $\pi'_{\tau(h)}(\cdot|x) \gets (1 - \epsilon) \pi_{\tau(h)}(\cdot|x) + \frac{\epsilon}{|A(x)|}$\;
    }
    \Else{
        $\pi'_{\tau(h)}(\cdot|x) \gets \pi_{\tau(h)}(\cdot|x)$\;
    }
    Sample action $a \sim \pi'_{\tau(h)}(a|x)$\;
    \uIf{$\tau(h)=i$}{
        $q_i[h] \gets 0$\;
        $q_i[h, a'] \gets 0$ for all $a'\in A(x)$\;
        \For{$a'\in A(x)$}{
            \If{$a'=a$}{
                $q_i[h, a'] \gets \textsc{TraverseOS}\left(ha', i, \pi'_{i}(a'|x)\cdot \rho_{i}\right) + u_i(h, a')$\;
            }
            $q_i[h, a'] \gets \frac{q_i[h, a']}{\pi_i'(a' | h)} + \mu d_i^{\pi, \sigma}(h, a')$\;
            $\widetilde{v}_{i}^{\pi, \sigma}(x, a') \gets \frac{q_i[h,a']}{ \rho_{i}}$\;
            $q_i[h] \gets q_i[h] + \pi_{i}(a'|x)\cdot q_i[h, a']$\;
        }
        \Return $q_i[h]$\;
    }
    \Else{
        \Return $\textsc{TraverseOS}\left(ha, i, \rho_{i}\right) + u_{i}(h,a)$\;
    }
}
\Hline{}
$\textsc{TraverseOS}(\emptyset, i, 1)$\;
\Return $\widetilde{v}_i^{\pi, \sigma}$\;
\end{algorithm}
    \Description[ValueEstimateOS]{ValueEstimateOS}
\end{figure}

\section{Exploitability under full game-tree traversals}
\label{sec:full game-tree traversal}

Figures~\ref{fig:last iterate full} and \ref{fig:average iterate full} present the exploitability of each algorithm's last- and average-iterate strategies under full game-tree traversals. 
In the last-iterate sense, PFTRL-RKL+ significantly outperforms PFTRL-KL+ in Leduc poker and shows a slight advantage in Goofspiel (4 cards) and Liars Dice (6 dice sides). However, in most other games, the differences between PFTRL-RKL+ and PFTRL-KL+ are relatively small, though PFTRL-KL+ is likely better. 
In the average-iterate sense, the perturbed FTRL algorithms are generally outperformed by the non-perturbed FTRL across most games. Nonetheless, in Leduc poker and Goofspiel (5 cards), FTRL diverges and is subsequently outperformed by the perturbed FTRL algorithms. Notably, even PFTRL-KL+ diverges in Leduc poker, allowing PFTRL-RKL+ to establish clear superiority over the other algorithms.

\begin{figure*}
    \centering
    \includegraphics[width=0.9\linewidth]{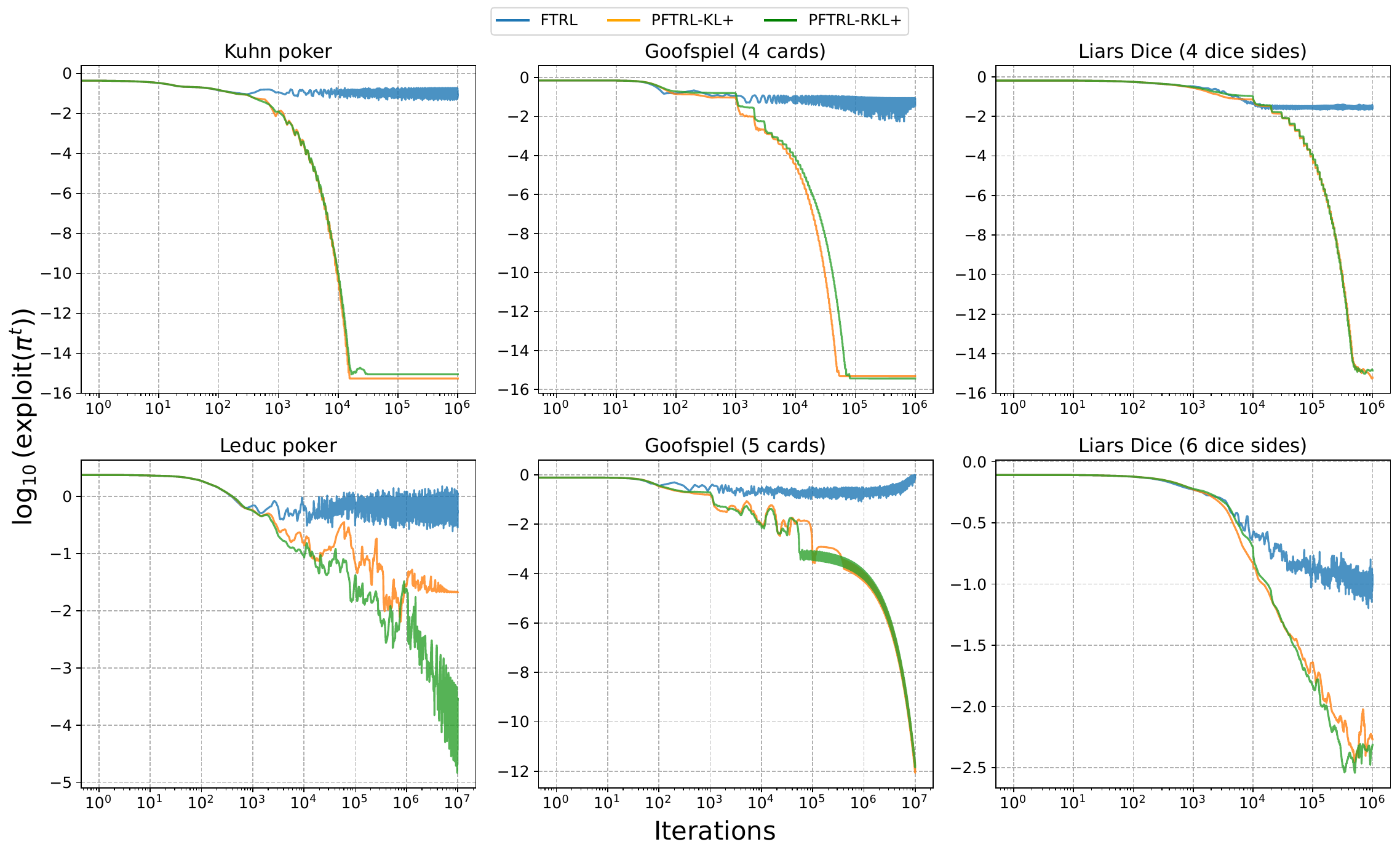}
    \caption{Exploitability of last iterate $\pi^t$ under full game-tree traversals.}
    \label{fig:last iterate full}
    \Description[Exploitability of last iterate $\pi^t$ with full game-tree traversals.]{Exploitability of last iterate $\pi^t$ with full game-tree traversals.}
\end{figure*}

\begin{figure*}
    \centering
    \includegraphics[width=0.9\linewidth]{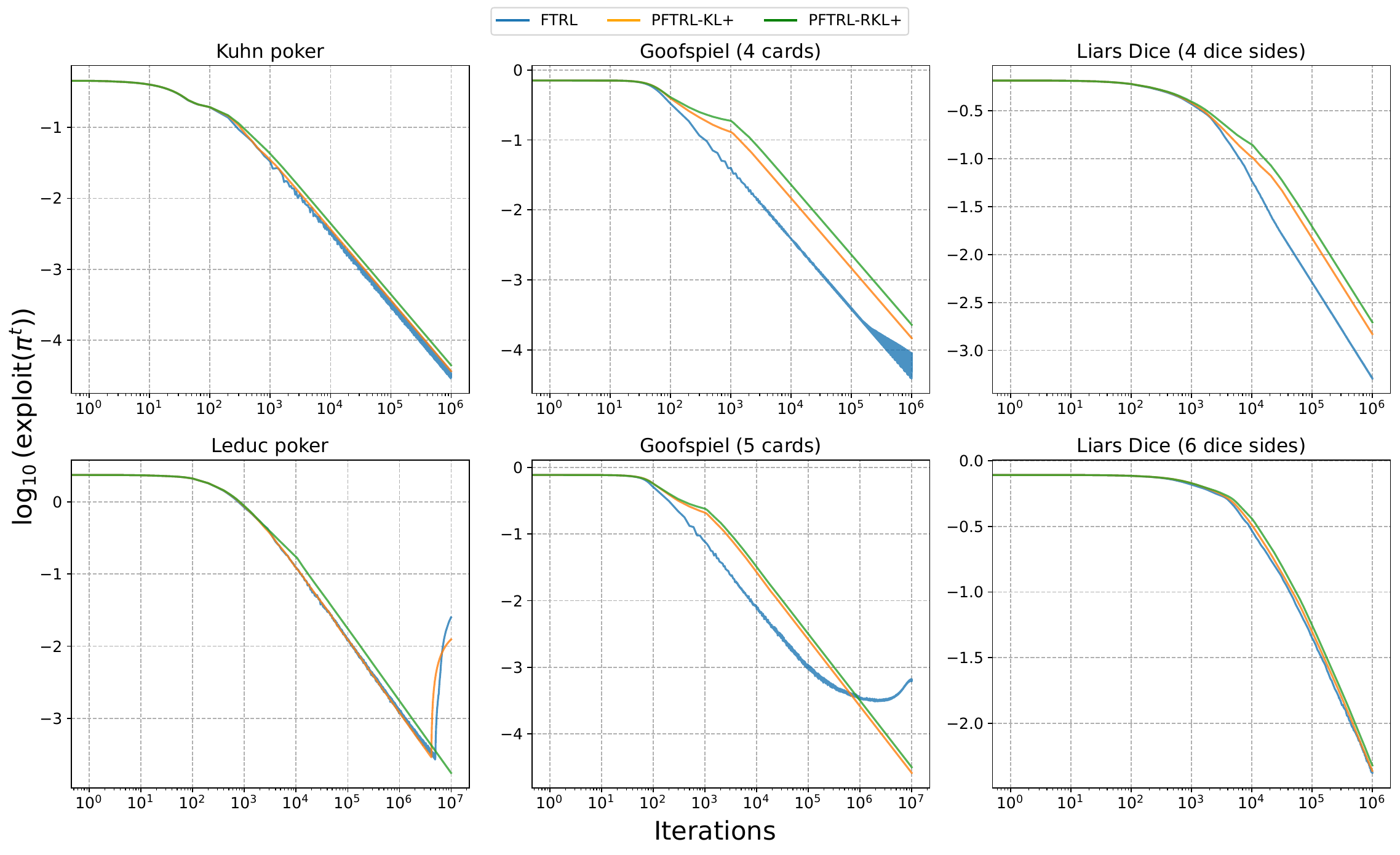}
    \caption{Exploitability of average iterate $\bar{\pi}^t$ under full game-tree traversals.}
    \label{fig:average iterate full}
    \Description[Exploitability of average iterate under a full game traverse]{Exploitability of average iterate under a full game traverse}
\end{figure*}

\section{Exploitability under external sampling}
\label{sec:external sampling}

Figure~\ref{fig:last iterate ES} depicts the exploitability of last-iterate $\pi^t$ under external sampling (ES), which samples only the opponent’s (and chance’s) choices, requiring a forward model of the game to recursively traverse all subtrees under the player’s actions. External sampling PFTRL-RKL+ (ES-PFTRL-RKL+) and -KL+ significantly outperform ES-FTRL in improving convergence and reducing exploitability. 
Notably, ES-PFTRL-RKL+ outperforms -KL+ in Leduc poker and slightly outperforms in Liars Dice (6 dice sides). 
While ES-PFTRL-KL+ edges out ES-PFTRL-RKL+ in the other games, ES-PFTRL-RKL+ consistently delivers stable results with smooth convergence behavior as well as under outcome sampling. 
Figure~\ref{fig:average iterate ES} depicts the exploitability of average-iterate $\bar{\pi}^t(x)$ under outcome sampling.
In contrast to the last-iterate case, ES-FTRL generally outperforms both ES-PFTRL-KL+ and -RKL+, achieving the lowest exploitability in most games except for Leduc poker, where ES-PFTRL-RKL+ has a noticeable advantage. 
While ES-PFTRL-RKL+ does not outperform KL+ in most games, its performance in Leduc poker highlights its potential in certain strategic settings.

\begin{figure*}
    \centering
    \includegraphics[width=0.9\linewidth]{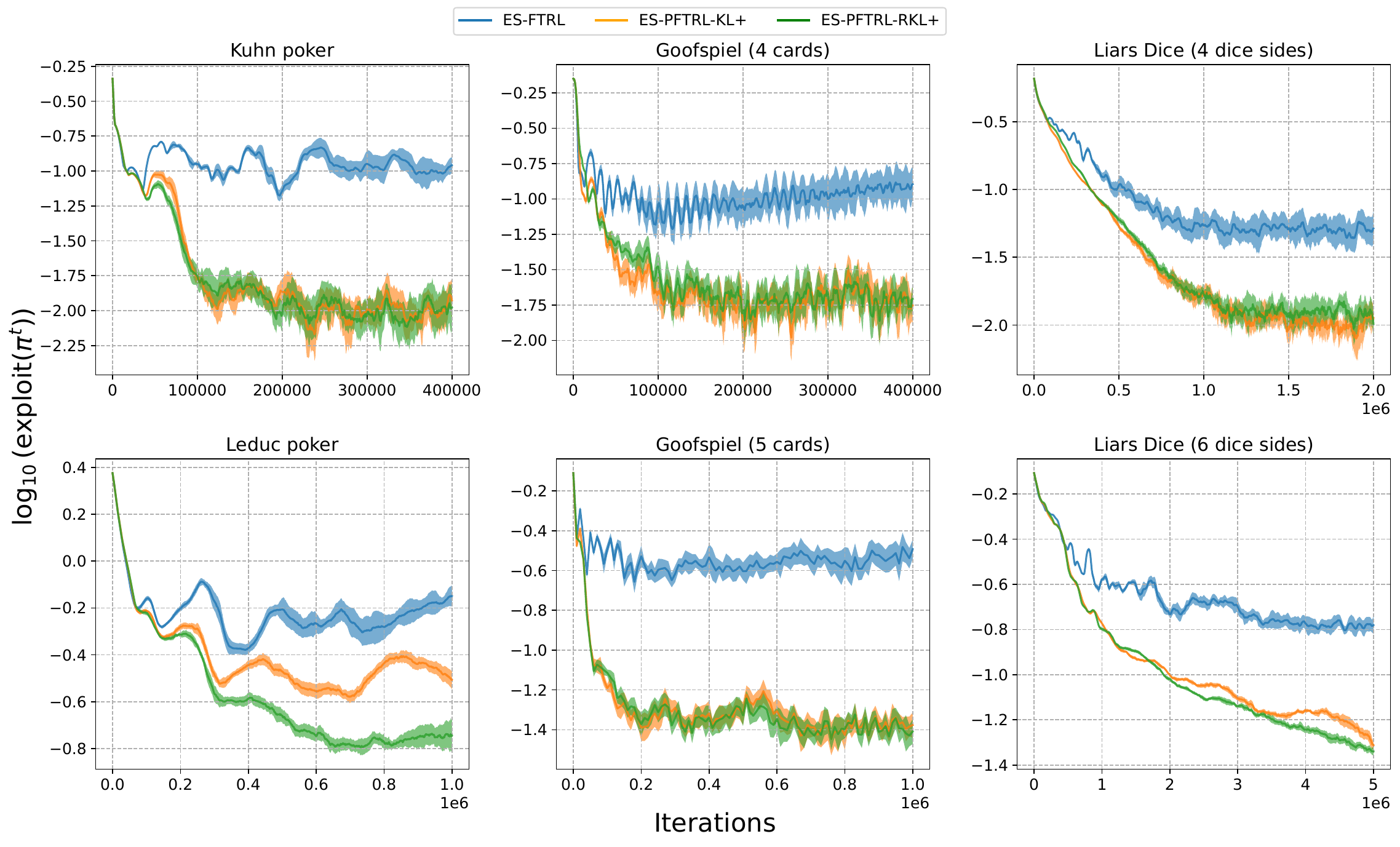}
    \caption{Exploitability of last iterate $\pi^t$ under external sampling.}
    \label{fig:last iterate ES}
    \Description[Exploitability of last iterate $\pi^t$ with external sampling.]{Exploitability of last iterate $\pi^t$ with external sampling.}
\end{figure*}

\begin{figure*}
    \centering
    \includegraphics[width=0.9\linewidth]{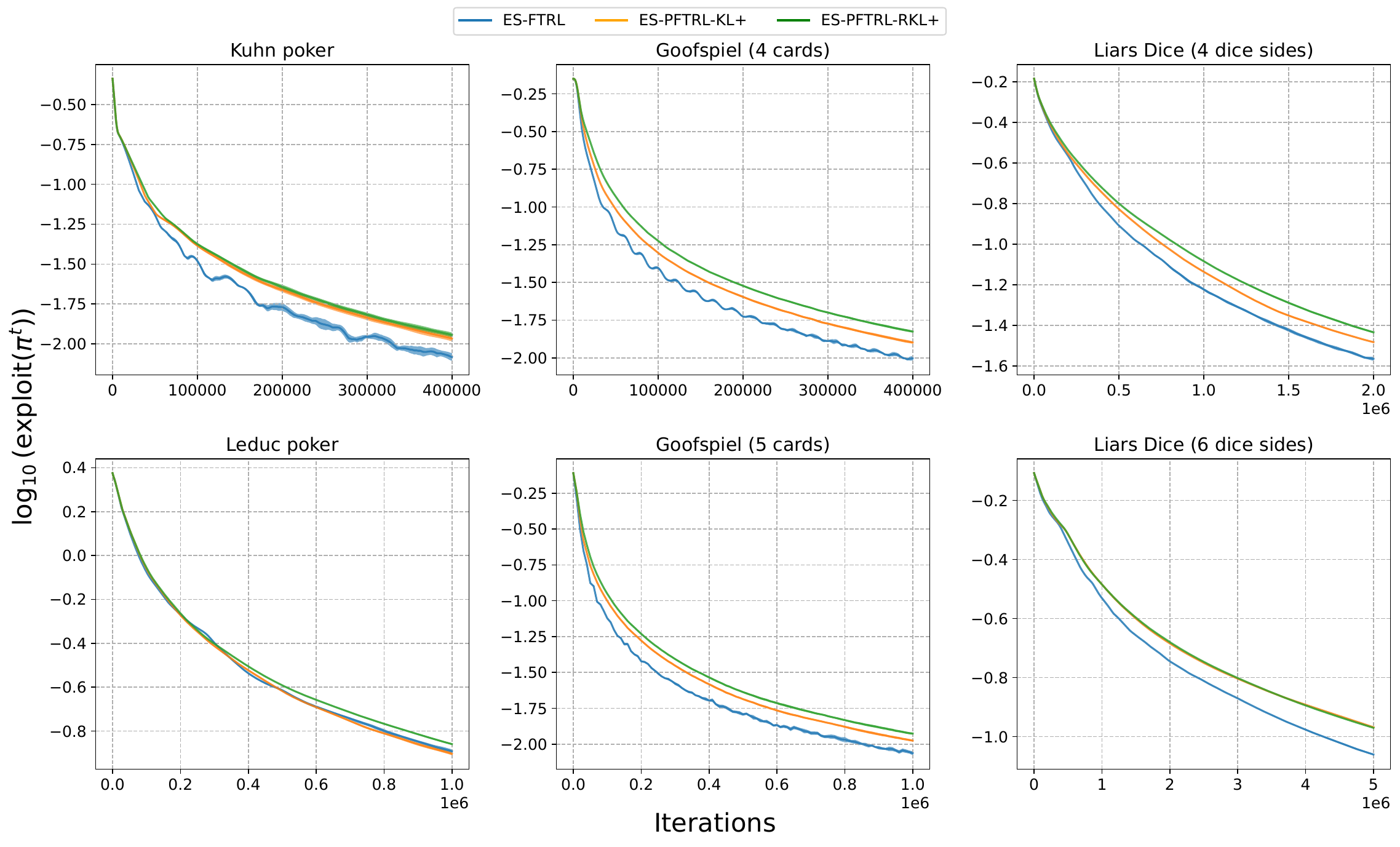}
    \caption{Exploitability of average iterate $\bar{\pi}^t$ under external sampling.}
    \label{fig:average iterate ES}
    \Description[Exploitability of average iterate under external sampling.]{Exploitability of average iterate under external sampling.}
\end{figure*}

\section{Comparison with CFR/CFR+ under outcome sampling}
\label{sec:comparison with cfr/cfr+}

We have added the results for CFR and CFR+ to Figure~\ref{fig:exploitability-last-iterate} in the main text and illustrate them in Figures~\ref{fig:last iterate OS} and \ref{fig:average iterate OS}.
The CFR-based algorithms were implemented following \cite{zinkevich2007regret,tammelin2014solving}. As shown in Figure~\ref{fig:average iterate OS}, CFR and CFR+ consistently outperform the other algorithms in the average-iterate sense, in line with previous studies. Even in Leduc poker, where the perturbed FTRL-based algorithms achieve the best performance in the last-iterate sense, CFR and CFR+ still demonstrate superior performance in the average-iterate sense.
In contrast, Figure~\ref{fig:last iterate OS} illustrates that CFR is consistently outperformed by the perturbed FTRL-based algorithms (OS-PFTRL-RKL+ and -KL+), as expected. However, CFR+ (OS-CFR+) shows significant performance across games. Notably, in Leduc poker and Liars Dice, OS-CFR+ either outperforms or matches the performance of OS-PFTRL-KL+ and OS-PFTRL-RKL+. These findings suggest an intriguing connection between CFR+ and perturbation-based approaches, potentially opening a new avenue for research into their interplay and performance dynamics.

\begin{figure*}
    \centering
    \includegraphics[width=0.9\linewidth]{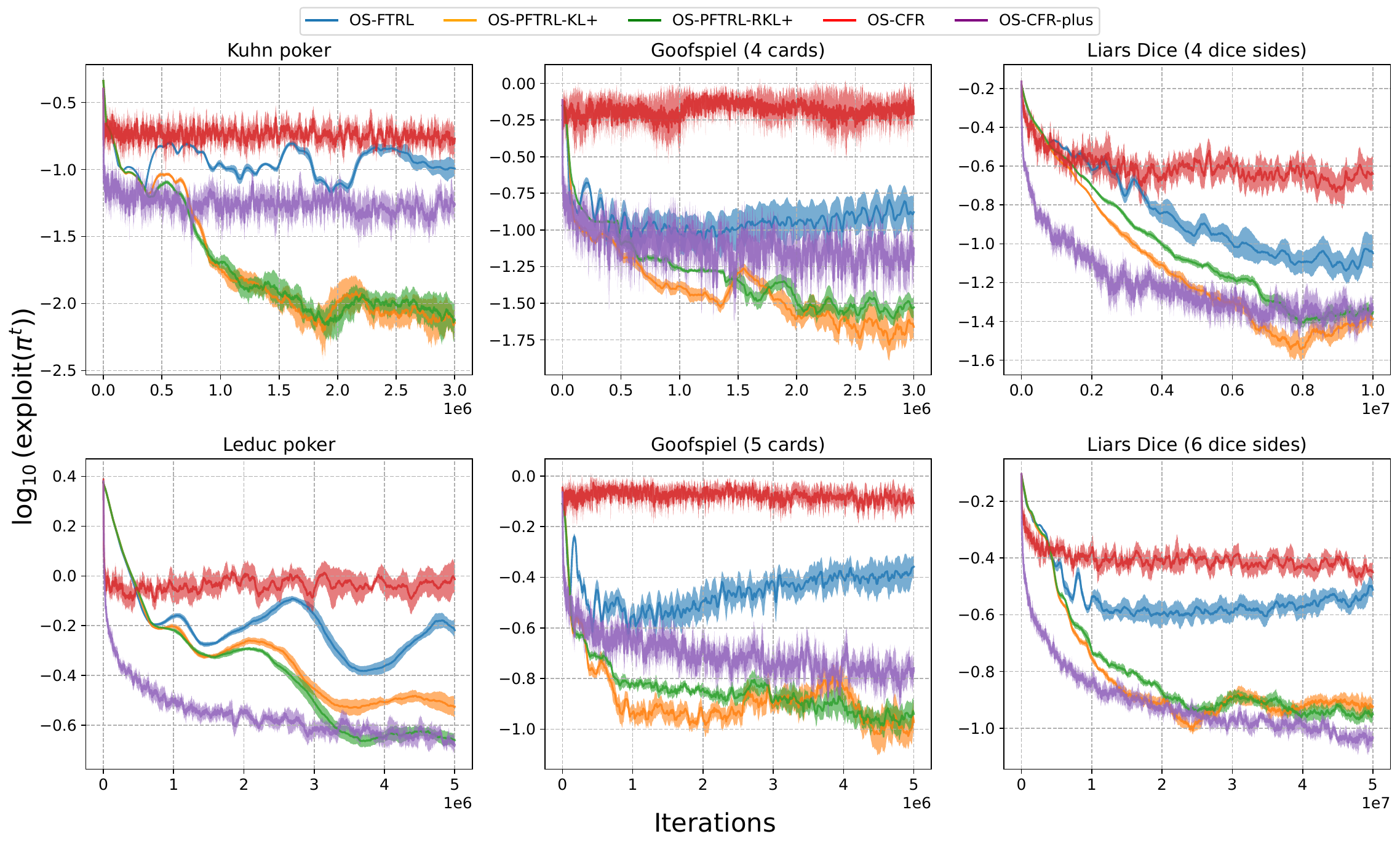}
    \caption{Exploitability of last iterate $\pi^t$ under outcome sampling.}
    \label{fig:last iterate OS}
    \Description[Exploitability of last iterate $\pi^t$ with outcome sampling.]{Exploitability of last iterate $\pi^t$ with outcome sampling.}
\end{figure*}

\begin{figure*}
    \centering
    \includegraphics[width=0.9\linewidth]{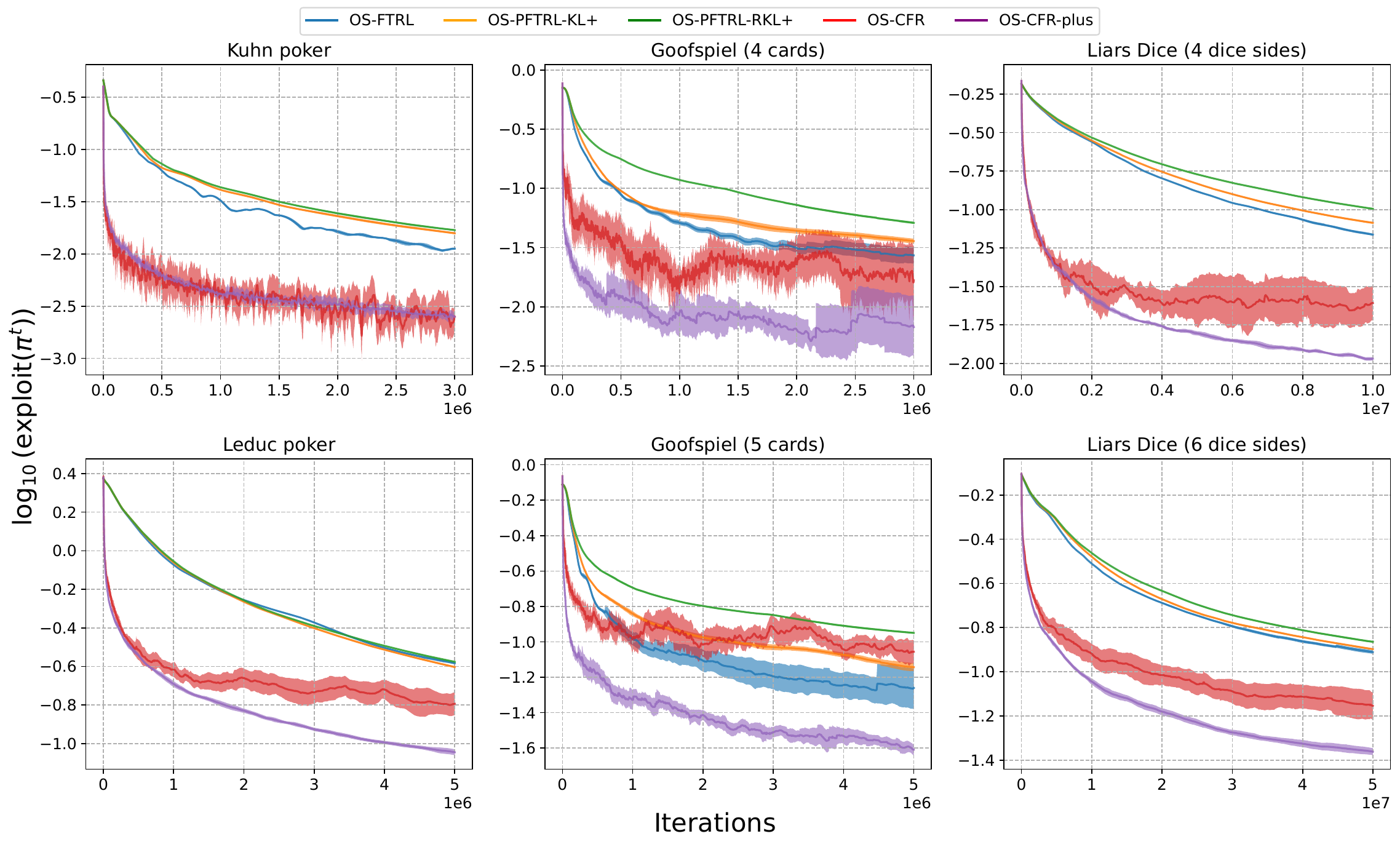}
    \caption{Exploitability of average iterate $\bar{\pi}^t$ under outcome sampling.}
    \label{fig:average iterate OS}
    \Description[Exploitability of average iterate under outcome sampling.]{Exploitability of average iterate under outcome sampling.}
\end{figure*}

\section{Comparison with Perturbed CFR+ under outcome sampling}
\label{sec:comparison with perturbed cfr+}

We have shown the results for Perturbed CFR+-L2+ (PCFR+-L2+) to Figures~\ref{fig:last iterate OS with PCFR-plus-L2+} and \ref{fig:average iterate OS with PCFR-plus-L2+}, compared with PCFR+-L2+ and the perturbed FTRL-based algorithms (OS-PFTRL-RKL+ and -KL+). 
PCFR+-L2+ is implemented following \cite{meng::2023}, employing the L2 distance to modulate the perturbation magnitude. 

We set the perturbation strength $\mu$ for perturbed CFR+ separately for each game: $\mu = 0.5$ in Kuhn poker and Goofspiel; $\mu = 0.1$ in Leduc poker; $\mu = 0.05$ in Liar's Dice with 4 dice sides; and $\mu =0.01$ in Liar's Dice with 6 dice sides.
We initialize the anchoring strategy 
uniformly: 
$\sigma_{i}(\cdot | x) = (1/|A(x)|)_{a \in A(x)}$ for each information set $x$.
It is updated every $T_{\sigma}=10,000$ visits under outcome sampling in all experiments. The sampling strategy is uniform sampling~\cite{mcaleer:iclr:2023}. The exploitability is averaged across $10$ random seeds for each algorithm and is presented on a logarithmic scale.

As shown in Figure~\ref{fig:average iterate OS with PCFR-plus-L2+}, 
the non-perturbed CFR+ consistently outperform the other algorithms, including 
PCFR+-L2+, in the average-iterate sense. That is, as we will see in FTRL, perturbation does not improve the performance. 
In contrast, in the last-iterate sense, Figure~\ref{fig:last iterate OS with PCFR-plus-L2+} illustrates that the non-perturbed CFR+ is consistently outperformed by PCFR+-L2+, FTRL-RKL+, and -KL+, as expected. 
PCFR+-L2+ shows significant performance across some of the games. 
Notably, in Leduc poker and Liars Dice, PCFR+-L2+ either outperforms or matches the performance of PFTRL-KL+ and PFTRL-RKL+. 
These findings suggest an intriguing connection between CFR+ and perturbation-based approaches, potentially opening a new avenue for research into their interplay and performance dynamics.

\begin{figure}
    \centering
    \includegraphics[width=0.9\linewidth]{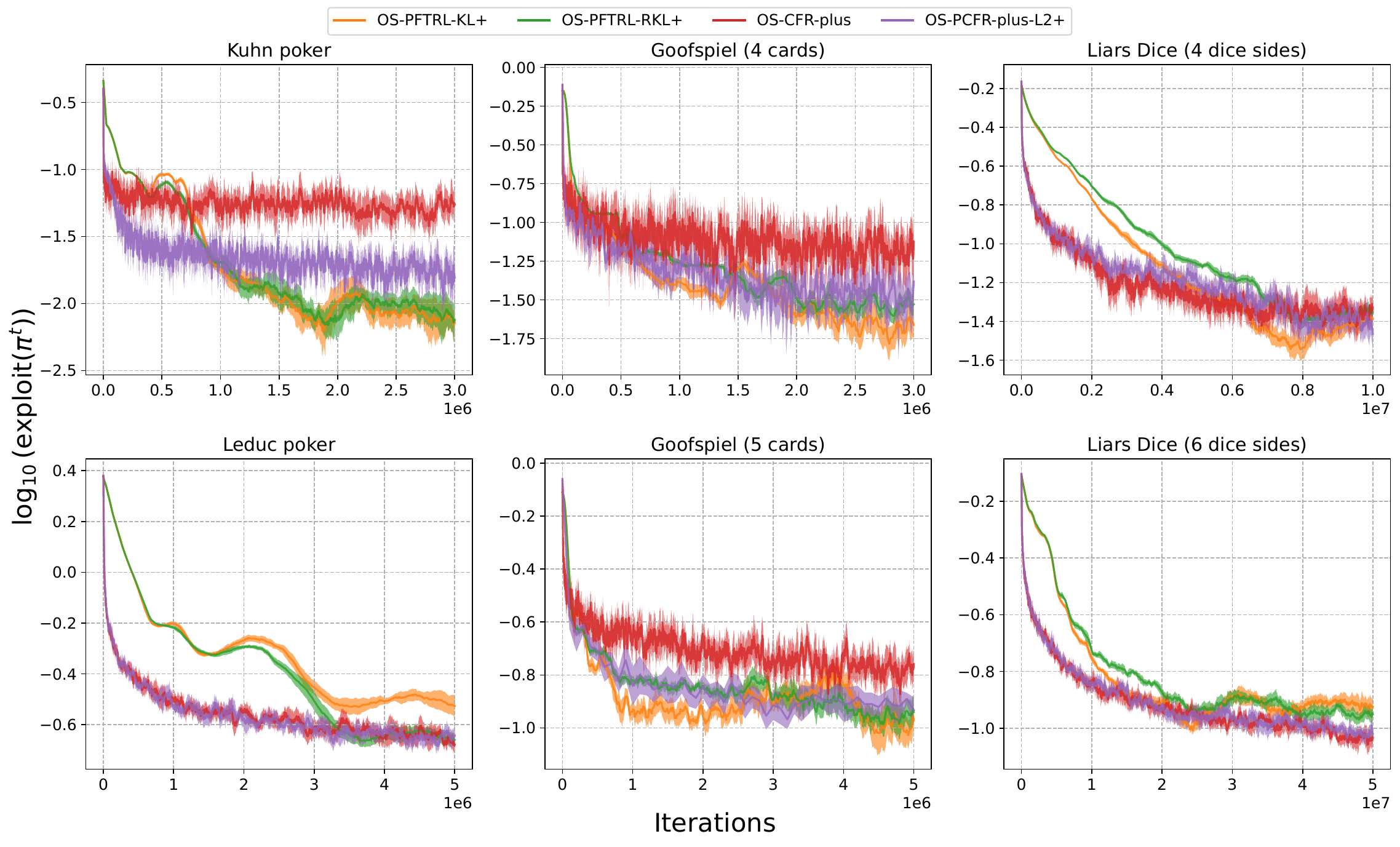}
    \caption{Exploitability of last iterate $\pi^t$ under outcome sampling with PCFR-plus-L2+.}
    \label{fig:last iterate OS with PCFR-plus-L2+}
\end{figure}
\begin{figure}
    \centering
    \includegraphics[width=0.9\linewidth]{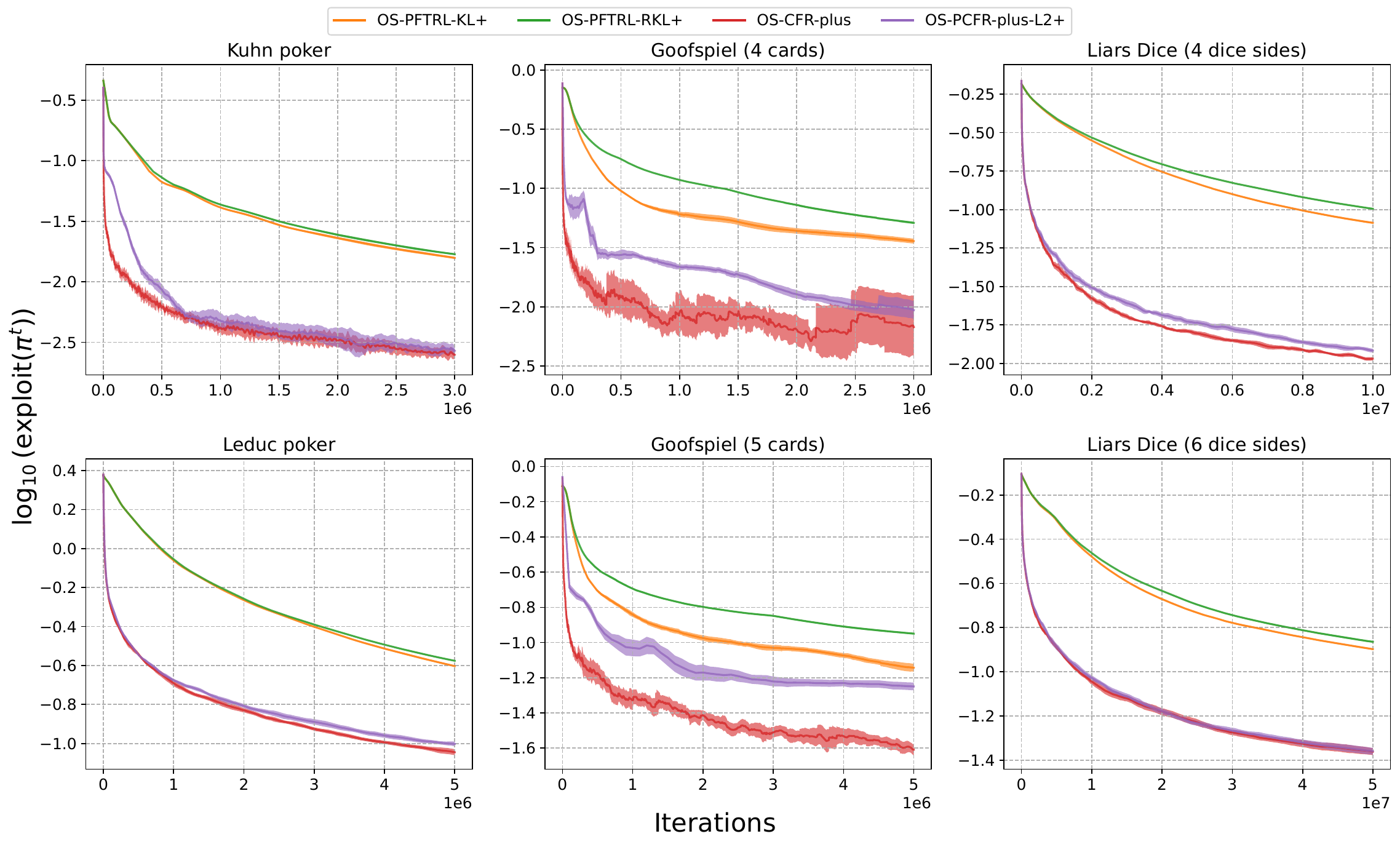}
    \caption{Exploitability of average iterate $\bar{\pi}^t$ under outcome sampling with PCFR-plus-L2+.}
    \label{fig:average iterate OS with PCFR-plus-L2+}
\end{figure}

\section{Anchoring strategy updates}
\label{sec:anchoring strategy updates}

This section discusses the anchoring strategy, which is an important component that specifies the magnitude of perturbation. Payoff perturbation 
introduces strongly convex penalties to the players' payoff functions to stabilize learning. Merely perturbing the payoffs results in the strategy converging only to an approximate Nash equilibrium. 
Fortunately, as highlighted in \cite{perolat2021poincare,abe2023last,abe:icml:2024}, 
they propose the {\it anchoring strategy update} approach, which ensures the updated strategy converges to an exact Nash equilibrium. The magnitude of perturbation is calculated as the product of a strongly convex penalty function and a perturbation strength parameter. 
Note that \citeauthor{liu2022power}~\citeyear{liu2022power}
shrink the perturbation strength at each iteration based on the current strategy profile's proximity to an underlying equilibrium so that it admits last-iterate convergence. However, it becomes challenging to choose an appropriate learning rate for the shrinking perturbation strength.

We implement their approach~\cite{perolat2021poincare,abe2023last,abe:icml:2024} into perturbed FTRL under sampling.
Specifically, every time the information set $x$ is visited $T_{\sigma} \leq T$ times, we update, or replaced the anchoring strategy $\sigma_i(\cdot | x)$ for $x$ with the current strategy $\pi_i^t(\cdot | x)$. 
Empirically, as Figure~\ref{fig:effect of anchoring strategy updates} in the Supplementary Material demonstrates, the anchoring strategy update improves PFTRL-RKL and -KL in the last- and average-iterate sense. We refer the algorithms with the anchoring strategy update as to PFTRL-RKL+ and -KL+, respectively. 
This process is illustrated in lines 10-14 of Algorithm~\ref{alg:ftrl_perturbed} in Supplementary Material. 


\begin{figure}[tb]
    \centering
     \includegraphics[width=0.6\linewidth]{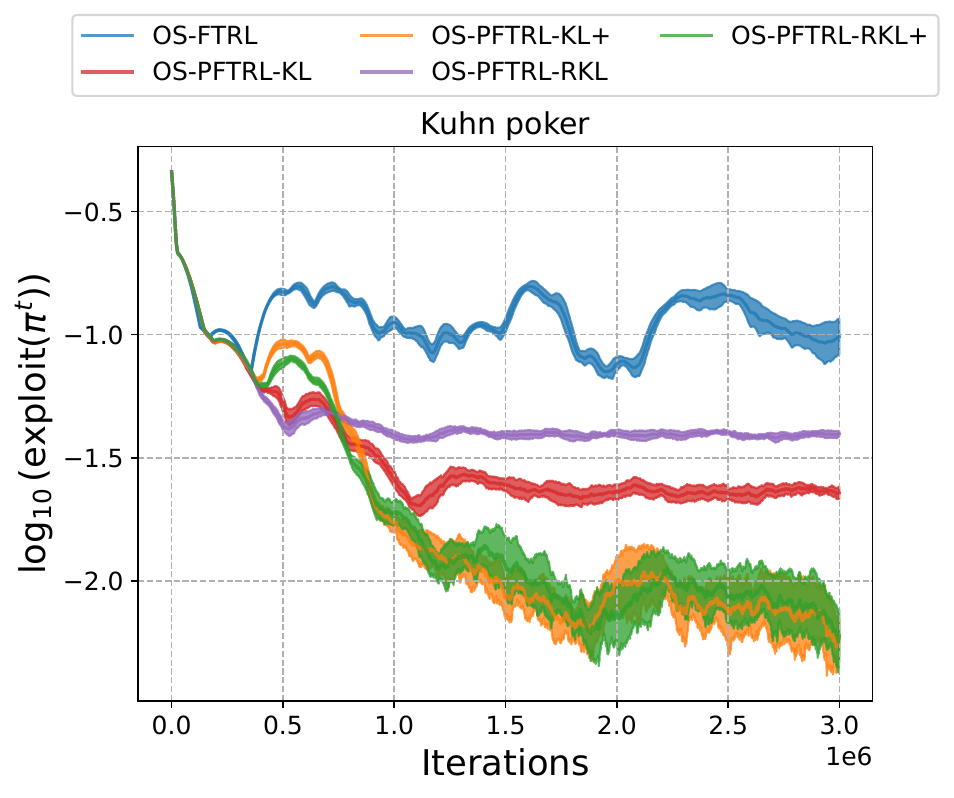}
    \caption{Effect of anchoring strategy updates of last-iterate in Kuhn poker.}
    \label{fig:effect of anchoring strategy updates}
    \Description[Effect of anchoring strategy updates]{Effect of anchoring strategy updates}
\end{figure}

Figure~\ref{fig:effect of anchoring strategy updates} illustrates the impact of anchoring strategy updates on the exploitability of the last iterate in Kuhn poker. The x-axis represents the number of iterations, while the y-axis indicates the level of exploitability. Over time, perturbations using KL or RKL divergences (OS-PFTRL-KL and -RKL) help the strategies converge faster to lower exploitability  compared to OS-FTRL. The ``+''variants using anchoring strategy updates (OS-PFTRL-KL+ and -RKL+) exhibit the best performance in the long run (lowest exploitability by the end of the iterations, dropping below~$-2.0$).

\end{document}